\documentclass[doc,floatsintext,apacite]{apa6}\usepackage[]{graphicx}\usepackage[]{xcolor}
\makeatletter
\def\maxwidth{ %
  \ifdim\Gin@nat@width>\linewidth
    \linewidth
  \else
    \Gin@nat@width
  \fi
}
\makeatother

\definecolor{fgcolor}{rgb}{0.345, 0.345, 0.345}

\usepackage{framed}
\makeatletter
 {\par\unskip\endMakeFramed%
 \at@end@of@kframe}
\makeatother

\definecolor{shadecolor}{rgb}{.97, .97, .97}
\definecolor{messagecolor}{rgb}{0, 0, 0}
\definecolor{warningcolor}{rgb}{1, 0, 1}
\definecolor{errorcolor}{rgb}{1, 0, 0}
\newenvironment{knitrout}{}{} 

\usepackage{alltt}
\usepackage{epsfig,amsmath,amsthm,amsfonts,alltt,bm,url,subcaption}
\usepackage[english]{babel}

\newtheorem{prop}{Proposition}
\newtheorem{remark}{Remark}

\title{Identification and Scaling of Latent Variables in Ordinal Factor Analysis}
\author{Edgar C. Merkle \and Sonja D. Winter \and Ellen Fitzsimmons}
\affiliation{University of Missouri}

\abstract{Social science researchers are generally accustomed to treating ordinal variables as though they are continuous. In this paper, we consider how identification constraints in ordinal factor analysis can mimic the treatment of ordinal variables as continuous. We specifically describe model constraints that lead to latent variable predictions equaling the average of ordinal variables. This result leads us to propose minimal identification constraints, which we call {\em integer constraints}, that place the latent variables on the scale of the observed, integer-coded ordinal variables. The integer constraints lead to intuitive model parameterizations because researchers are already accustomed to thinking about ordinal variables as though they are continuous. We provide a proof that our proposed integer constraints are indeed minimal identification constraints, as well as illustrations of how integer constraints work with real data. We also provide simulation results indicating that integer constraints are similar to other identification constraints in terms of estimation convergence and admissibility.}

\authornote{  Correspondence to Edgar Merkle, 219 McAlester Hall, Columbia, Missouri, USA.
  Email: merklee@missouri.edu\\ \ \\ Financial support: This work was supported by the Institute of Education Sciences, U.S. Department of Education, Grant R305D210044. \\ \ \\
  Competing interests: The authors declare none.
}
\shorttitle{Identification of Ordinal CFA}

\let\proglang=\textsf
\let\pkg=\emph

\IfFileExists{upquote.sty}{\usepackage{upquote}}{}
\begin{document}
\maketitle

In factor analysis and related models of ordinal observed variables, we commonly assume that latent variables follow a normal distribution with mean 0 and variance 1. These constraints have computational advantages that can lead to efficiency in model estimation. Separately from identification constraints, it is common practice for applied researchers to ignore that their observed variables are ordinal, summing or averaging the variables as though they are continuous \cite<e.g.,>{lidkru18,sij24}. In this paper, we propose identification constraints that are related to averaging ordinal variables as though they are continuous. This can make the model parameters more intuitive to applied researchers, as compared to the usual identification constraints.

Many researchers have studied when and whether we can treat ordinal variables as continuous \cite<e.g.,>{bolbar81,burvuo19,lidkru18,mcn20,rhebro12,winmar84}. Perhaps the most famous work on this topic is Stevens' scales of measurement \cite<e.g.,>{ste46}. In distinguishing between ordinal scales and interval scales, Stevens notes that ``means and standard deviations computed on an ordinal scale are in error to the extent that the successive intervals on the scale are unequal in size'' (p.\ 679). Our results below involve the idea of equal intervals in ordinal CFA models, providing minimal identification constraints that are related to equal intervals. Our results are also related to those of Kruschke and colleagues \cite{kru14,kru15,lidkru18}, who considered identification constraints for univariate, ordinal regression models. They reasoned that, because applied researchers are accustomed to treating ordinal variables as continuous, we should seek to identify the ordinal regression model so that the underlying continuous variable is on the scale of the ordinal variable. For example, if we have an ordinal variable with 5 categories, then the ordinal regression model should generally predict values between 1 and 5 on the latent continuous scale, which are then converted to probabilities of assuming each ordered category.

In the pages below, we formalize the above arguments by first providing background on the specific models and identification constraints that we consider. We then study how ordinal CFA models can be constrained so that the latent variable predictions equal the average of the ordinal variables (where we treat the ordinal variables as continuous). Next, we propose minimal identification constraints related to these ideas and illustrate them via example and simulation. Finally, we consider limitations and future directions. The supplementary material includes code showing how our proposed integer constraints can be implemented in {\em lavaan} \cite{ros12} and in {\em mirt} \cite{cha12}.

\section{Theoretical Background}
We assume data vectors $\bm{y}_i$ of length $p$, $i=1,\ldots,n$, where all $p$ variables are ordinal with $K$ categories. Under the traditional probit link function, we can conceptualize continuous, latent data vectors $\bm{y}^\ast_i$ that are chopped to yield the observed, ordinal data. For example, for $K=4$, the chopping can be written as
\begin{align*}
y_{ij} = 1 &\text{ if }-\infty <\ y^*_{ij} < \tau_{j1} \\
y_{ij} = 2 &\text{ if }\tau_{j1} <\ y^*_{ij} < \tau_{j2} \\
y_{ij} = 3 &\text{ if }\tau_{j2} <\ y^*_{ij} < \tau_{j3} \\
y_{ij} = 4 &\text{ if }\tau_{j3} <\ y_{ij}^* <\ \infty,
\end{align*}
where $\tau_{j1} < \tau_{j2} < \tau_{j3}$ are the threshold parameters for item $j$.

The CFA model is placed on the $\bm{y}^\ast_i$ as if we had observed, continuous data:
\begin{align}
    \bm{y}^\ast_i &= \bm{\nu} + \bm{\Lambda} \bm{\eta}_i + \bm{\delta}_i \\
    \bm{\eta}_i &\sim \text{N}(\bm{\kappa}, \bm{\Phi}) \label{eq:lvdist} \\
    \bm{\delta}_i &\sim \text{N}(\bm{0}, \bm{\Theta})
\end{align}
where $\bm{\nu}$ is $p \times 1$, $\bm{\Lambda}$ is $p \times m$, $\bm{\eta}_i$ is $m \times 1$, and $\bm{\delta}_i$ is $p \times 1$. We further assume that $\bm{\Theta}$ is diagonal and that $\bm{\Lambda}$ has a {\em clustered} structure, i.e., that each observed variable only loads on one factor. Regarding the latter assumption, we could alternatively say that the factor complexity of each observed variable equals 1 or that each row of $\bm{\Lambda}$ has only one nonzero entry.

Given $\bm{\eta}_i$, the probability that $Y_{ij}$ assumes each category is the area of the normal distribution between two thresholds, i.e.,
\begin{equation}
  \label{eq:condlik}
P(Y_{ij} = y_{ij} \mid \bm{\eta}_i, \bm{\xi}) = \Phi \left (\frac{\tau_{j,y_{ij}} - (\nu_j + \bm{J}_j \bm{\Lambda} \bm{\eta}_i)}{\theta_{ii}} \right ) - \Phi \left (\frac{\tau_{j,(y_{ij}-1)} - (\nu_j + \bm{J}_j \bm{\Lambda} \bm{\eta}_i)}{\theta_{ii}} \right ),
\end{equation}
where $\Phi()$ is the standard normal cumulative distribution function, $\bm{J}_j$ is a $1 \times p$ vector with an entry of 1 in position $j$ and 0 elsewhere, $\bm{\xi}$ is a vector of item parameters, and $\tau_{j0} = -\infty$ and $\tau_{jK} = \infty$ for all $j$. The conditional model likelihood for respondent $i$ (conditioned on the latent variables $\bm{\eta}_i$) can then be written as:
\begin{equation}
  \label{eq:lik}
L(\bm{\xi} | \bm{y}_i, \bm{\eta}_i) = \prod_{j=1}^p \prod_{k=1}^K P(Y_{ij} = k \mid \bm{\eta}_i, \bm{\xi})^{u_{ijk}},
\end{equation}
where $u_{ijk}$ equals 1 if person $i$ responded to question $j$ with the $k$th ordered category and 0 otherwise. For model estimation, the marginal likelihood is often used instead of the above likelihood, where the $\bm{\eta}_i$ are integrated out. This integration requires approximation via quadrature or other numerical methods \cite<e.g.,>{tuerij06}. Alternatively, researchers often obtain the polychoric correlations between ordinal variables and fit the traditional CFA model via weighted least squares \cite<e.g.,>{mut84}. The latter approach is fast because it avoids numerical integration, capitalizing on the equivalence between IRT and CFA \cite<e.g.,>{takdel87}.

\subsection{Identification Constraints}
Additional constraints are necessary to identify model parameters. For example, a common set of constraints are:
\begin{equation}
    \label{eq:tradid}
\text{diag}(\bm{\Phi}) = \bm{1},\ \bm{\kappa} = \bm{0},\ \bm{\nu} = \bm{0},\ \bm{\Theta} = \bm{I},
\end{equation}
where the restriction on $\bm{\Phi}$ is sometimes called a ``unit variance constraint.'' A variation involves fixing one loading per latent variable to 1, instead of fixing each diagonal entry of $\bm{\Phi}$ to be 1. This shifts the constraints on $\bm{\Phi}$ to constraints on $\bm{\Lambda}$, and is sometimes called a ``reference indicator constraint.'' Another variation for ordinal CFA involves the so-called ``delta parameterization,'' where the constraints on $\bm{\Theta}$ are replaced with constraints on the model-implied covariance matrix of $\bm{y}^\ast$:
\begin{equation}
\text{diag}(\bm{\Lambda \Phi \Lambda}^\prime + \bm{\Theta}) = \bm{1}.
\end{equation}
These sets of constraints lead to equivalent, equal-fitting models whose parameter estimates can be transformed to one another. While the specific choice of constraints is often regarded as arbitrary \cite<e.g.,>{bollil24}, it is worth mentioning that different sets of constraints sometimes lead to differing conclusions regarding parameter equality \cite{kloklo18,klo23,ste02} and regarding Bayesian model selection \cite{gramer22}.

\subsection{Latent Variable Prediction}
Following model estimation via marginal maximum likelihood or weighted least squares, researchers may optionally request latent variable predictions that serve as scores for each individual $i$. There is a large history of literature discussing the indeterminacy of factor scores \cite<see, e.g.,>[for a summary]{wal22}, where the indeterminacy is discussed in the context of estimating the $\bm{\eta}_i$ jointly with the $\bm{\delta}_i$. To obtain unique predictions of the $\bm{\eta}_i$, a reasonable thing to do (which is also common practice) is to marginalize over the $\bm{\delta}_i$ while addressing sign indeterminacy and rotational indeterminacy via parameter constraints. This is similar to the situation that \citeA{rhesav25} recently considered for continuous $\bm{y}_i$.

For ordinal factor analysis, we can obtain latent variable predictions by maximizing the likelihood function $L(\bm{\eta}_i \mid \bm{y}_i, \bm{\xi})$ for all $i$, where the likelihood function has the same form as the right side of Equation~\eqref{eq:lik}. As compared to Equation~\eqref{eq:lik}, we now estimate $\bm{\eta}_i$ and condition on $\bm{\xi}$ whereas we previously did the opposite. Maximization of this function requires numerical methods because it involves the normal CDF.

Maximum likelihood estimates of the $\bm{\eta}_i$ do not exist for extreme response patterns consisting of all 1s or $K$s. Consequently, it is common practice (for IRT as well as generalized linear mixed models) to multiply the likelihood function by the ``prior'' distribution from Equation~\eqref{eq:lvdist}, which leads us to maximize the posterior distribution of each $\bm{\eta}_i$. The resulting estimates of the $\bm{\eta}_i$ are called the {\em maximum a posteriori} (MAP) estimates. In situations where we have already estimated the Equation~\eqref{eq:lvdist} parameters and hold them fixed, we may also refer to our estimates of the $\bm{\eta}_i$ as {\em empirical Bayes} estimates. Further detail about these procedures can be found in, e.g., \citeA{bakkim04}.

\section{Parameter Constraints and Sum Scores}
It is customary for applied researchers to ignore the fact that their variables are ordinal and to sum or average the ordinal variables associated with each latent variable. This commonly happens by assigning the lowest category a value of 1 and the highest category a value of $K$, then averaging. The average serves as a summary score for each participant that can be used in regressions and other models. We now discuss how the latent variable predictions from an ordinal CFA model can mimic the average of observed ordinal variables. This will lead us to develop alternative identification constraints in later sections.

\subsection{Constraints}
Consider the ordinal CFA model from the previous section, where all free loadings are fixed at 1, $\bm{\kappa} = (\frac{K + 1}{2}) \bm{1}$, where $\bm{1}$ is an $m \times 1$ vector, and $( \tau_{j1}, \tau_{j2}, \ldots, \tau_{j(K-1)}) = (1.5,\ 2.5,\ \ldots,\ (K - .5))$. 
Under these constraints, we have a Rasch-like model, and the items are interchangeable because the loadings and thresholds are identical across items. For such a model, \citeA{and77} shows that the sum of individual $i$'s responses is a sufficient statistic for $\bm{\eta}_i$ \cite<also see>{and78,lord53}. \citeA{sam69} additionally shows that the maximum of the item response function for response category $k$ occurs at the midpoint between that category's threshold parameters (see her Equation 5.6), for $k = 2, \ldots, (K-1)$. Our restrictions on thresholds imply that the mode occurs at the integer value that applied researchers often assign to ordinal variables. Thus, we claim that the MAP estimates of the $\bm{\eta}_i$ are equal to the average of observed ordinal responses (where the responses are coded as integers starting from 1). But further clarification is needed for the extreme categories of 1 and $K$, which we provide in the next section.

\subsection{Empirical Results}
As described previously, the latent variable predictions involve maximization of Equation~\eqref{eq:lik}, which is now a function of $\bm{\eta}_i$ and is conditioned on $\bm{\xi}$ along with $\bm{y}_i$. 
To show that the model constraints from the previous section lead to latent variable predictions equaling the average of observed variables, we consider here a one-factor model with values of $p$ from 2 to 10 and $K=5$. For each value of $p$, we generate all possible response patterns and calculate the MAP prediction of the latent variable for each response pattern. We do not estimate any item parameters here: in addition to the constraints on the loadings, thresholds, and latent means from the previous subsection, we fixed $\bm{\nu} = \bm{0}$, $\bm{\Phi} = p$, and $\bm{\Theta} = \bm{I}$.

Figure~\ref{fig:mllv} shows scatter plots of the average observed response (x-axis) versus MAP latent variable prediction for all possible response patterns. Each red point is a response pattern that does not include an extreme response of 1 or 5, while each blue point is a response pattern that does include an extreme. The figure shows that the points generally fall along the diagonal, with some differences at the far left and far right side of each panel. This provides some evidence that latent variable predictions under our model constraints remain close to the average of observed variables for all response patterns. The supplementary materials include additional code that considers additional values of $p$ and $K$. It shows that the gradient of the likelihood function is always close to 0 at the mean of observed variables, so long as the response pattern does not include extreme responses of 1 or $K$. The code also considers maximum likelihood estimates of the latent variables, in addition to MAP predictions.

When an individual's response pattern does include the extremes of 1 or $K$, the $\bm{\eta}_i$ predictions are pulled toward $-\infty$ or $+\infty$, respectively, so that they no longer equal the average of the observed variables. This can be observed on the left and right sides of each panel of Figure~\ref{fig:mllv}. 
A similar phenomenon happens for large values of $K$ (say 8 or more) when responses are near the extremes (e.g., 2 or $(K-1)$).
In the MAP case, the prior distribution from~\eqref{eq:lvdist} helps keep the predictions from straying too far from the observed average. We fixed the prior variance, $\bm{\Phi}$, to equal $p$ in each panel of these results. This may appear to be an odd choice, but it is used here to demonstrate robustness of our result. This is because a prior variance of $p$ is weaker than the traditional prior variance of 1. Were we to fix $\bm{\Phi}$ to 1, our points would be even closer to the diagonal.
And in the maximum likelihood case, an adhoc, vague prior distribution is often used to ensure that latent variable predictions exist for extreme response patterns. In each case, the resulting latent variable predictions are close to the means of the integer-coded ordinal variables.

\begin{figure}
  \caption{Observed averages versus MAP latent variable estimates for $K=5$ and $p=2$ to 10. Each point represents a response pattern. Red points are response patterns that do not include a response of 1 or 5, and blue points are response patterns that do include a response of 1 and/or 5.}
  \label{fig:mllv}
\begin{knitrout}\footnotesize
\definecolor{shadecolor}{rgb}{0.969, 0.969, 0.969}\color{fgcolor}

{\centering \includegraphics[width=4in,height=4in]{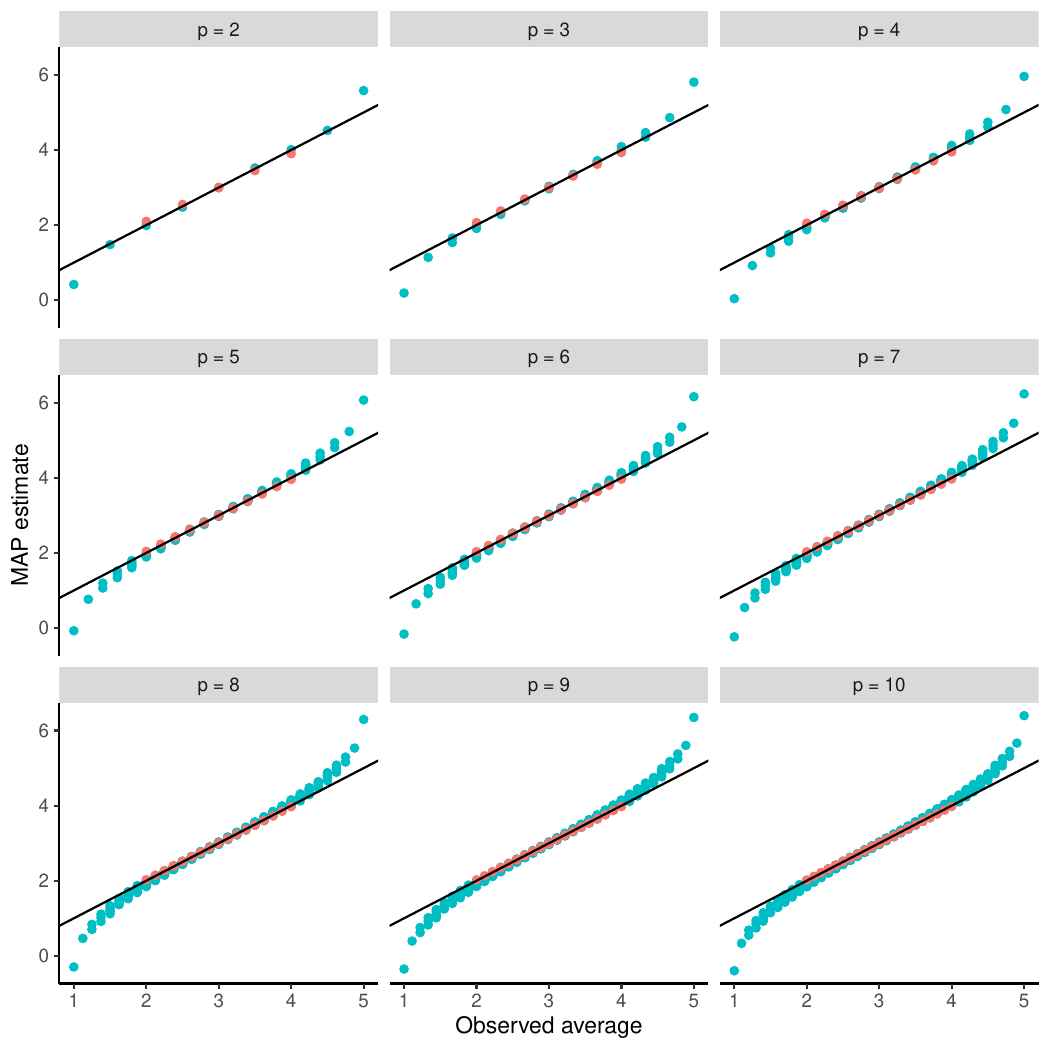} 

}

\end{knitrout}
\end{figure}

\subsection{Summary}
We have shown that under a highly-constrained ordinal CFA model, the MAP predictions of the latent variables are highly related to the integer-coded average of observed variables. This was anticipated by \citeA{and77} and \citeA{sam69}, though the connection to treating ordinal variables as continuous was perhaps not fully clarified or appreciated. For example, about twenty years after these works, \citeA{ste94} states, ``My strong hunch is that, if scales were developed using unit weighting on the basis of ordinary component analysis, and these scale scores were used instead of individual items, that there would be no need for special techniques for categorical variables, because the resulting scores would be `close enough' to continuous variates'' (p.\ 218). Our result is also similar to results of \citeA{folgro22}, who show that equally-spaced thresholds can lead integer-coded correlations to match polychoric correlations (especially see their Corollary 1).
But the model described in this section is too highly constrained to be useful in many practical situations, so we next consider minimal identification constraints.

\section{Minimal Identification Constraints in Ordinal CFA}
Although researchers nearly always identify ordinal CFA models via some variation of the constraints in~\eqref{eq:tradid}, there exist an infinite number of possible identification constraints. We would like a set of identification constraints that get us closer to the highly-constrained model from the previous section, where latent variable predictions are averages of observed variables. 

To move in this direction, we consider the \citeA{wuest16} matrix expressions that transform parameter estimates under one set of constraints to parameter estimates under another set of constraints. Their expressions are
\begin{align}
\label{eq:tran1}
  \widetilde{\bm{T}} &= \bm{\gamma 1}^\prime + \bm{\Delta}^{-1}\bm{T} \\
\label{eq:tran2}
  \widetilde{\bm{\Lambda}} &= \bm{\Delta}^{-1}\bm{\Lambda D} \\
\label{eq:tran3}
  \widetilde{\bm{\nu}} &= \bm{\Delta}^{-1}\bm{\nu} + \bm{\Delta}^{-1}\bm{\Lambda \beta} + \bm{\gamma} \\
\label{eq:tran4}
  \widetilde{\bm{\Theta}} &= \bm{\Delta}^{-1} \bm{\Theta} \bm{\Delta}^{-1} \\
\label{eq:tran5}
  \widetilde{\bm{\kappa}} &= \bm{D}^{-1} (\bm{\kappa} - \bm{\beta}) \\
  \label{eq:tran6}
  \widetilde{\bm{\Phi}} &= \bm{D}^{-1} \bm{\Phi} \bm{D}^{-1},
\end{align}
where $\bm{T}$ is a $p \times (K-1)$ matrix whose rows each contain the thresholds for one observed variable, and $\bm{D}$, $\bm{\Delta}$, $\bm{\beta}$, and $\bm{\gamma}$ are the transformation matrices and vectors. The $\bm{D}$ and $\bm{\Delta}$ matrices are positive, diagonal matrices of dimension $m \times m$ and $p \times p$, respectively. The $\bm{\beta}$ and $\bm{\gamma}$ vectors are of dimension $m \times 1$ and $p \times 1$, respectively.

Ordinal CFA parameter identification amounts to defining a minimal set of parameter constraints that fix the four transformation matrices and vectors described above, such that the constraints hold on both the left and right sides of Equations~\eqref{eq:tran1}--\eqref{eq:tran6} \cite<also see>[Proposition 1]{wuest16}. For example, consider the identification constraints from~\eqref{eq:tradid}. These constraints require that $\bm{D} = \bm{I}$, $\bm{\beta} = \bm{0}$, $\bm{\gamma} = \bm{0}$, and $\bm{\Delta} = \bm{I}$. Below, we use the transformation matrices to develop alternative constraints.

\section{Alternative Identification Constraints}
Instead of fixing parameters to 0 or 1, we seek identification constraints that put the latent variable close to the integer scale of the ordinal variable. As we mentioned earlier, such constraints can be helpful to applied researchers who are working with ordinal data, because they are accustomed to thinking on the scale of the ordinal variable and to treating the ordinal variables as if they are continuous.

The constraints that we study are related to the constraints that led to factor scores mimicking observed averages. Instead of fixing $\bm{\nu}$ to $\bm{0}$, we require that the $\nu$ parameters associated with each latent variable sum to 0. Relatedly, instead of fixing a single loading to 1 or fixing the latent variance to 1, we constrain the loadings associated with each latent variable to average 1. This is reminiscent of the \citeA{little06} effect coding approach for continuous data. Finally, we fix the lower and upper thresholds of each observed variable to 1.5 and $K - 0.5$, respectively.

To formally describe the constraints, let $\mathcal{S}_q$ be the set of observed variables whose loadings in the $q$th column of $\Lambda$ are not fixed to 0 (i.e., the set of observed variables that ``load'' on latent variable $q$). Let $n_q$ be the cardinality of $\mathcal{S}_q$. Then our identification constraints can be written as
  \begin{align*}
    \displaystyle \sum_{j \in \mathcal{S}_q} \nu_j = 0\ \ &\forall\ q = 1, \ldots, m \\
    \frac{1}{n_q} \displaystyle \sum_{j \in \mathcal{S}_q} \lambda_{jq} = 1\ \ &\forall\ q = 1, \ldots, m \\
    \tau_{j1} = 1.5\ \ &\forall\ j = 1, \ldots, p \\
    \tau_{j(K-1)} = K - 0.5\ \ &\forall\ j = 1, \ldots, p.
  \end{align*}
  
Based on our previous arguments, these threshold restrictions help ensure that the scale of each latent variable is similar to that of the integer-coded ordinal items. Additionally, the latent variable means and variances are freely estimated, reflecting the standing of each latent variable on the ordinal scale. This helps applied researchers to understand and interpret the latent variable predictions, as well as other model parameters.

To show that the above constraints are minimal identification constraints, we first note that we have $2(p + m)$ individual constraints, which matches the number that was established by \citeA{wuest16}. To further establish these constraints, we make use of the Wu and Estabrook transformation matrices in the following proposition.

\begin{prop}
Let $\mathcal{S}_q$ be the set of observed variables whose loadings in the $q$th column of $\Lambda$ are not fixed to 0. Let $n_q$ be the cardinality of $\mathcal{S}_q$. Then the following are minimal identification constraints for an ordinal CFA model with clustered structure:
  \begin{align*}
    \displaystyle \sum_{j \in \mathcal{S}_q} \nu_j = 0\ \ &\forall\ q = 1, \ldots, m \\
    \frac{1}{n_q} \displaystyle \sum_{j \in \mathcal{S}_q} \lambda_{jq} = 1\ \ &\forall\ q = 1, \ldots, m \\
    \tau_{j1} = 1.5\ \ &\forall\ j = 1, \ldots, p \\
    \tau_{j(K-1)} = K - 0.5\ \ &\forall\ j = 1, \ldots, p.
  \end{align*}
\end{prop}

\begin{proof}
  By Proposition 1 of \citeA{wuest16}, we first show that the proposed constraints fix the transformation matrices from Equations~\eqref{eq:tran1}--\eqref{eq:tran6}. We then show that these constraints do not add additional parameter restrictions.

  For a particular observed variable $j$, the right side of Equation~\eqref{eq:tran1} involves scaling its thresholds by a positive constant $\delta_{jj}$ and then adding a constant $\gamma_j$. Considering these transformations, we must set $\gamma_j=0$ and $\delta_{jj}=1$ to maintain lower and upper threshold values of 1.5 and $(K-.5)$, respectively. This holds for all $j$, so we have $\bm{\gamma} = \bm{0}$ and $\bm{\Delta} = \bm{I}$.

  Next, we examine~\eqref{eq:tran2} with $\bm{\Delta} = \bm{I}$. The right side of this equation scales each column of $\bm{\Lambda}$ by a positive, diagonal entry of $\bm{D}$. But we already constrained the free entries in each column of $\bm{\Lambda}$ to average 1. The only way to maintain this constraint is to set $\bm{D} = \bm{I}$.

  Finally, we examine~\eqref{eq:tran3} with $\bm{\Delta} = \bm{I}$ and $\bm{\gamma} = \bm{0}$ and consider a particular latent variable $q$. To maintain the requirement that $\displaystyle \sum_{j \in \mathcal{S}_q} \widetilde{\nu}_j = 0$, we require that
\begin{equation*}
  \displaystyle \sum_{j \in \mathcal{S}_q} \lambda_{jq} \beta_q = - \displaystyle \sum_{j \in \mathcal{S}_q} \nu_j.
\end{equation*}
But we also have the restriction that $\displaystyle \sum_{j \in \mathcal{S}_q} \nu_j = 0$. So we must fix $\beta_q = 0$ for all $q$, i.e., $\bm{\beta} = \bm{0}$.
Now all four transformation matrices are fixed, establishing that these constraints resolve parameter indeterminacy.

  To show that the proposed constraints are minimal constraints required to identify the model parameters, we note that parameters identified under traditional constraints can be transformed to the proposed constraints. This is achieved via the following set of transformation matrices.
\begin{align*}
  \bm{\Delta}_{jj} &= (\tau_{j(K-1)} - \tau_{j1})/(K - 2) \ \ \forall\ j \\
  \bm{D}_{kk} &= n_k \left ( \sum_{j \in \mathcal{S}_k} \delta^{-1}_{jj} \lambda_{jk} \right )^{-1} \ \ \forall\ k \\
  \bm{\beta}_k &= -\left ( \displaystyle \sum_{j \in \mathcal{S}_k} \delta^{-1}_{jj} \lambda_{jk} \right )^{-1} \sum_{j \in \mathcal{S}_k} \left ( 1.5 + \delta^{-1}_{jj}(\nu_j - \tau_{j1}) \right )  \ \ \forall\ k \\
  \bm{\gamma}_j &= 1.5 - \delta^{-1}_{jj} \tau_{j1} \ \ \forall\ j.
\end{align*}
  
\end{proof}

The identification constraints proposed here are not the only ones that could be used. Following tradition, we could fix one loading per latent variable instead of requiring that loadings average 1. We could also add constraints on $\bm{\nu}$ and/or on $\bm{\Theta}$ and reduce the constraints on thresholds. We further discuss some of these alternatives in Appendix A. Our focal constraints appear to lead to the closest correspondence between integer-coded averages and latent variable predictions.

\citeA{lee90} discuss ideas related to our proposed constraints, identifying ordinal CFA models via constraints on thresholds \cite<also see>{lee07,shi98}. However, they do not consider the idea of placing the latent variables on the scale of the ordinal variables. In their example, they fix some thresholds to the maximum likelihood estimates of a previous study, where those estimates come from a model whose latent variables follow a standard normal distribution.

We now discuss some additional issues related to our proposed constraints.
\begin{remark}
The proposed identification constraints are minimial identification constraints. This means that, as compared to traditional identification constraints, the model fit and many other model summaries remain the same. In particular, standardized coefficients under the proposed constraints are equal to those obtained under traditional constraints.
\end{remark}
Remark 1 is especially noteworthy because some researchers are accustomed to reporting standardized coefficients. The proposed constraints have no impact on standardized coefficients, and it remains precarious to compare estimated coefficients across groups, standardized or otherwise. For example, although the latent variable means and variances are free under integer constraints, some of the thresholds are held equal across groups. Additionally, because we are not changing the fit of the model, model misfit and model misspecification are concerns for models with our proposed constraints, just as they are for models with traditional constraints. For example, \citeA{grofol24} recently considered how assumed normality of the $\bm{y}_i^\ast$ can bias the polychoric correlations that are used for weighted least squares estimation.

\begin{remark}
  To convert parameter estimates under alternative constraints (e.g., those from Proposition 1) to parameter estimates under the traditional constraints from Equation~\eqref{eq:tradid}, the transformation matrices are
  \begin{align*}
    \text{diag}(\bm{D}) &= \text{diag}(\bm{\Phi})^{1/2} \\
    \text{diag}(\bm{\Delta}) &= \text{diag}(\bm{I}) \\
    \bm{\beta} &= \bm{\kappa} \\
    \bm{\gamma} &= - \bm{\Theta}^{-1} (\bm{\nu} + \bm{\Lambda \kappa})
   \end{align*}
\end{remark}
This result is similar to the results of \citeA{klo21} for models of continuous variables, except that they are for models of ordinal variables.

In summary, Proposition 1 establishes that our proposed constraints address the model's parameter indeterminacy without introducing further restrictions. In the sections below, we first study whether the constraints cause problems with convergence of model estimation algorithms. We then illustrate how the proposed constraints work in two applied examples.

\section{Simulation Study}
We used a Monte Carlo simulation to ascertain that the proposed identification constraints do not affect model convergence, admissibility, or quality (as defined by the value of the model discrepancy function at the optimal estimates). We fit a variety of ordinal factor analysis models in \pkg{lavaan} using default options, to examine whether researchers using integer constraints are likely to encounter problems with model estimation.

\subsection{Method}
In the simulation study, we compared the proposed integer constraints to reference-marker constraints and to unit-variance constraints using a population model with three correlated factors. We varied attributes that are often included in latent variable simulation designs \cite<e.g.,>{gagne2006measurement, flora2004empirical, rhebro12}: number of indicators per factor (3 or 6), standardized factor loading magnitude (.4, .6, .8), number of response categories (3, 4, 5), response distribution (symmetric, skewed, or middling). In the skewed conditions, the response probability of the highest option was .04 (and in conditions with > 2 response options, the response probability of the second highest option was set to .06). In the middling conditions, the response probability of the lowest and highest response options were .05 (where this condition was not included for 2 response options). For conditions with sparse response distributions, we manipulated the proportion of indicators per latent factor affected by that sparse pattern (.33, .66, 1). For proportions less than 1, the remaining items had a symmetric response distribution.

In addition to these population model conditions, we also compared the two starting value options offered by {\em lavaan}: simple and default. With simple starting values, all parameter values are set to zero, except the factor loadings, which are set to 0.7, and (residual) variances, which are set to one. The default starting values are more involved. First, the factor loadings are estimated per factor using a two stage least squares estimator. Second, the residual variances of observed variables are set to half the observed variance, and all other (residual) variances are set to 0.05. Third, thresholds are set to the standard normal distribution variates that match the (cumulative) response probabilities. The remaining parameters (regression coefficients, covariances) are set to zero. 

We used {\em lavaan} \cite{ros12} to simulate 500 datasets for each fully crossed condition. Next, we used {\em lavaan} to fit the ordinal CFA model to each dataset, using each of the three identification constraints. These estimations used the default {\em lavaan} three-stage DWLS algorithm with ``theta'' parameterization. The sum constraints involved in our integer coding are handled in {\em lavaan} by projecting the full parameter vector to a reduced vector with nonredundant entries, then estimating this reduced parameter vector. See \citeA{ros15} for further detail.

The simulation outcomes of interest were convergence rate, admissible results rates (e.g., non-negative variance estimates and positive definite covariance matrices), and $\chi^2$ model fit estimates. We evaluated the impact of the conditions with a fixed-effects ANOVA, focusing on the partial Eta-squared ($\eta^{2}_{p}$) estimates, which were computed using {\em effectsize} \cite{benshachar}.

\subsection{Results}
We did not find much evidence that the integer constraints had estimation differences as compared to alternative identification constraint methods. Minor differences in convergence rates existed, but these were balanced out by differences in admissible result rates, resulting in almost identical converged and admissible (i.e., valid) result rates. Results of an ANOVA with converged and admissible result rates as the outcome variable indicated that the identification constraint had a negligible effect ($\eta^{2}_{p} = 0.001$). Similarly, starting values also minimally affected converged and admissible result rates ($\eta^{2}_{p} = 0$). Other simulation factors had a larger impact, ranging from $\eta^{2}_{p} = 0.088$ for response distribution to $\eta^{2}_{p} = 0.238$ for factor loading magnitude. Given the minimal impact of starting values, we will focus on the results when using simple starting values. Results for default starting values are presented in Appendix B. 

\paragraph{Convergence by Condition.}
To provide further insight into these findings, we depict a subset of conditions in Figure~\ref{fig:simres}. Within this figure, the y-axis shows the proportion of replications that converged and were admissible. Different factor loading magnitudes are shown on the x-axis, panel rows represent the number of indicators per factor, and panel columns represent the number of response categories. Within each plot, the three identification constraints are defined by different shapes and colors, and different response distributions are separated by line type. For the skewed and middling response distributions, we included results in which all indicators follow this pattern. We focus on these conditions because we found that results increasingly resembled the symmetric response distribution as the proportion of indicators with the skewed or middling response distributions decreased. Thus, the results in Figure~\ref{fig:simres} represent the most challenging conditions.

\begin{figure}
  \caption{Proportion of converged and admissible replications across simulation conditions when all indicators have a balanced, skewed, or middling response distribution.}
  \label{fig:simres}
\begin{knitrout}\footnotesize
\definecolor{shadecolor}{rgb}{0.969, 0.969, 0.969}\color{fgcolor}

{\centering \includegraphics[width=5in,height=5in]{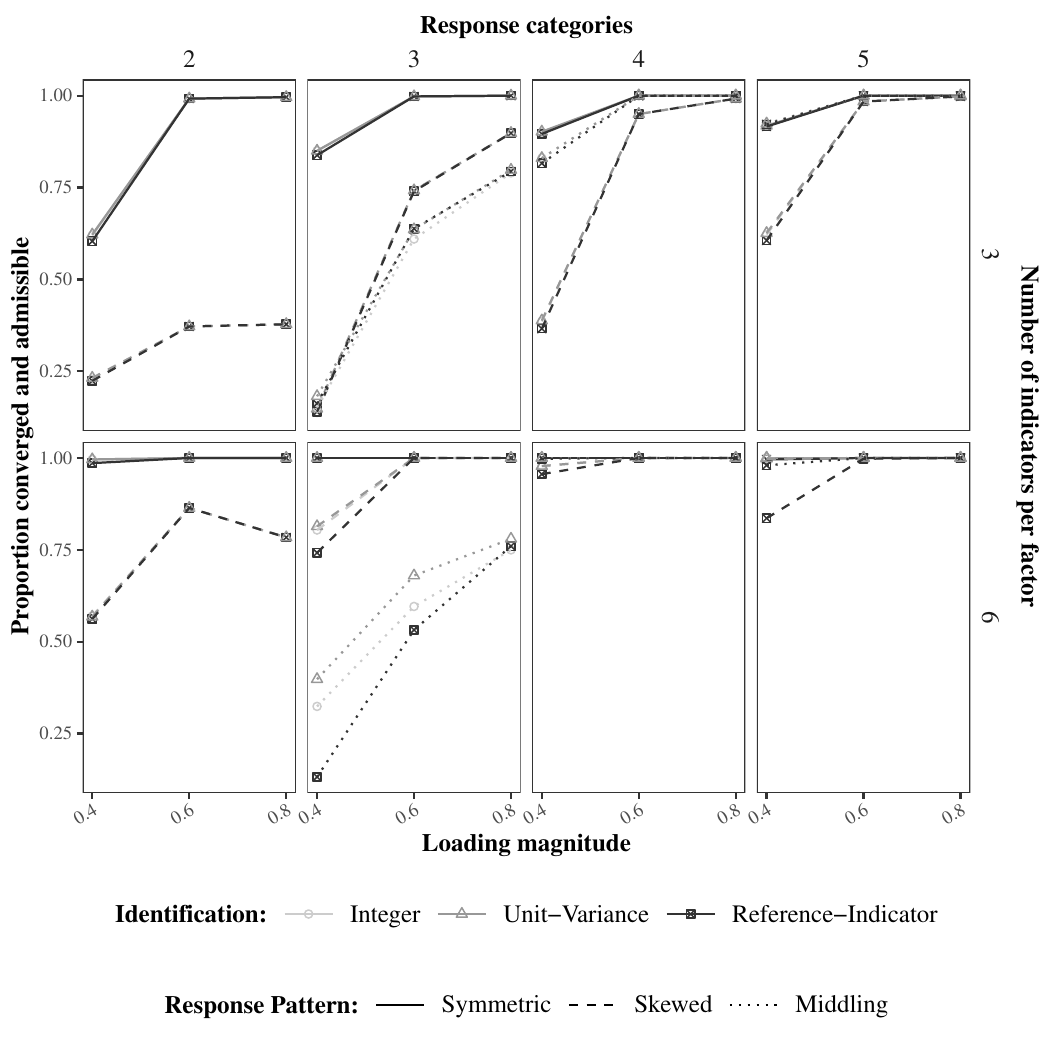} 

}

\end{knitrout}
\end{figure}

Figure~\ref{fig:simres} demonstrates that converged and admissible result rates are higher for models with more indicators, better measurement quality (i.e., higher factor loadings), items with more response categories, and symmetric response distributions. However, within a specific combination of these factors, the three identification constraint methods performed similarly (i.e., lines of matching type have near perfect overlap).

There were two exceptions to the finding that identification constraints performed similarly. These exceptions are both shown in row 2, column 1 of Figure~\ref{fig:simres}. First, for models with six indicators per factor, 0.4 factor loading magnitude, and items with three response categories which followed a skewed response distribution (dashed lines), the reference-indicator constraints resulted in lower converged and admissible result rates ($0.74$) compared to the other two identification constraint methods ($0.81$). Second, for the middling response distribution (dotted lines), model estimation was often problematic. For example, when the loading magnitude was 0.4, the reference-indicator constraints had a ``converged and admissible rate'' of $0.13$, with the proposed integer constraints having a rate of $0.32$ and the unit-variance constraints having a rate of $0.4$. These rates increase and become more similar as the loading magnitude increases. The conditions appear especially difficult because there are two thresholds per item, but nearly all the responses are in the middle category. This leads to considerable uncertainty in the thresholds, which is magnified by small loadings.

\paragraph{Estimation Quality.}
For those replications where all identification constraint methods converged and were admissible, we examined quality of estimation using the $\chi^2$ statistics of model fit (rounded to three decimal points). Similar to the convergence results from the previous paragraph, we observed differences in $\chi^2$ statistics under middling response distributions and three response categories (see Table~\ref{tab:chisq}). Differences decreased as the factor loading magnitude increased. A closer inspection of the differences in $\chi^2$-values across all conditions shows that the integer identification constraints most often resulted in a different $\chi^2$-value ($56$\%), followed by unit-variance ($20$\%), reference-marker ($19$\%), and replications where all three identification constraint methods produced different $\chi^2$-values ($6$\%). When response distributions were symmetric or skewed, $\chi^2$-values were identical for the vast majority of replications (see Appendix B).

To better understand the differences in $\chi^2$-values that occurred with the middling response distribution, we focused on the most problematic conditions with six indicators per factor that had three response options. Table~\ref{tab:chisq2} shows that, when differences across identification methods arose, the reference-marker identification method was somewhat more likely to result in the best fit (i.e.,lowest $\chi^2$-value), in some cases together with a second identification method. This pattern was more apparent when the middling response distribution was applied to all indicators and the loading magnitude was lowest. Full results for all middling response distribution conditions are included in Appendix B.

\begin{table}[ht]
\centering
\begingroup\footnotesize
\begin{tabular}{llcccccc}
  \hline
  Loading Mag. & Prop. Sparse &  \multicolumn{6}{c}{Response Options}\\ & & \multicolumn{3}{c}{3 Indicators} & \multicolumn{3}{c}{6 Indicators}\\ & & 3 & 4 & 5 & 3 & 4 & 5\\ \hline
0.4 & 0.33 & 0.99 & 1.00 & 1.00 & 0.87 & 0.99 & 0.99 \\ 
   & 0.67 & 0.94 & 1.00 & 1.00 & 0.65 & 0.98 & 0.99 \\ 
   & 1.00 & 0.93 & 1.00 & 1.00 & 0.51 & 0.99 & 0.99 \\ 
  0.6 & 0.33 & 1.00 & 1.00 & 1.00 & 0.92 & 1.00 & 1.00 \\ 
   & 0.67 & 0.98 & 1.00 & 1.00 & 0.75 & 0.99 & 1.00 \\ 
   & 1.00 & 0.93 & 1.00 & 1.00 & 0.56 & 1.00 & 1.00 \\ 
  0.8 & 0.33 & 1.00 & 1.00 & 1.00 & 0.99 & 1.00 & 1.00 \\ 
   & 0.67 & 1.00 & 1.00 & 1.00 & 0.97 & 1.00 & 1.00 \\ 
   & 1.00 & 1.00 & 1.00 & 1.00 & 0.95 & 1.00 & 1.00 \\ 
   \hline
\end{tabular}
\endgroup
\caption{Proportion replications with middling response pattern resulting in identical fit across identification constraint methods.} 
\label{tab:chisq}
\end{table}

\begin{table}[ht]
\centering
\begingroup\footnotesize
\begin{tabular}{lccccccccc}
  \hline
  Best Fit & \multicolumn{9}{c}{Proportion Sparse Indicators}\\ & \multicolumn{3}{c}{Loadings: 0.4} & \multicolumn{3}{c}{Loadings: 0.6} & \multicolumn{3}{c}{Loadings: 0.8}\\ & 0.33 & 0.67 & 1.00 & 0.33 & 0.67 & 1.00 & 0.33 & 0.67 & 1.00\\ \hline
All & 0.87 & 0.65 & 0.51 & 0.92 & 0.75 & 0.56 & 0.99 & 0.97 & 0.95 \\ 
  Reference-Indicator & 0.01 & 0.04 & 0.18 & 0.00 & 0.02 & 0.06 & 0.00 & 0.00 & 0.01 \\ 
  Unit-Variance & 0.03 & 0.06 & 0.02 & 0.01 & 0.05 & 0.04 & 0.00 & 0.01 & 0.01 \\ 
  Integer & 0.03 & 0.07 & 0.04 & 0.00 & 0.04 & 0.07 & 0.01 & 0.00 & 0.00 \\ 
  RI \& UV & 0.01 & 0.05 & 0.11 & 0.02 & 0.05 & 0.09 & 0.00 & 0.02 & 0.01 \\ 
  RI \& I & 0.01 & 0.05 & 0.07 & 0.02 & 0.04 & 0.08 & 0.00 & 0.00 & 0.02 \\ 
  UV \& I & 0.03 & 0.09 & 0.07 & 0.02 & 0.03 & 0.10 & 0.00 & 0.00 & 0.01 \\ 
   \hline
\end{tabular}
\endgroup
\caption{Proportion replications with middling response pattern, six indicators, and three response categories resulting in best fit across identification constraint methods.} 
\label{tab:chisq2}
\end{table}

\paragraph{Summary.}
The simulation study showed that the proposed integer identification constraints do not meaningfully affect estimation admissibility, convergence, or quality. When differences do emerge, the proposed integer identification constraints are more similar to the unit-variance identification constraint method, and both perform better than the reference-indicator identification method. Problems can arise when there are few ordinal categories, and the bulk of responses are in a single middle category. In this case, estimation is more difficult regardless of identification constraint, and integer coding does not necessarily perform best. But integer coding also does not consistently perform worse than other sets of constraints in those situations.

\begin{table}[ht]
\centering
\begingroup\footnotesize
\begin{tabular}{rccccccc}
  \hline
 & comfort & environment & work & future & technology & industry & benefit \\ 
  \hline
1 &   5 &  29 &  33 &  14 &  18 &  10 &  21 \\ 
  2 &  32 &  90 &  98 &  72 &  91 &  47 & 100 \\ 
  3 & 266 & 145 & 206 & 210 & 157 & 173 & 193 \\ 
  4 &  89 & 128 &  55 &  96 & 126 & 162 &  78 \\ 
   \hline
\end{tabular}
\endgroup
\caption{Item response frequencies of the attitudes toward science dataset.} 
\label{tab:datfreq}
\end{table}

\section{Example 1: Comparison to Traditional Estimates}
To build intuition for how the constraints work in practice, we use real data to compare a model with traditional identification constraints to a model with our proposed integer constraints. We use a 7-item survey of attitudes toward science and technology \cite{reif15}, where each item has the ordered categories of ``strongly disagree,'' ``disagree,'' ``agree,'' ``strongly agree.'' The dataset includes responses from 392 individuals, with no missing values. It is available via the {\em ltm} R package \cite{ltm}, with item response frequencies being shown in Table~\ref{tab:datfreq}.

\subsection{Method}
We used {\em lavaan} \cite{ros12} to fit a 1-factor, ordinal CFA model to the 7 items via the default DWLS algorithm (obtained via the argument ordered = TRUE). We first fit the model using the traditional constraints from Equation~\eqref{eq:tradid} (i.e., using the ``theta'' parameterization), and we then fit the model using the alternative constraints of:
\begin{align*}
  \displaystyle \sum_{j=1}^7 \nu_j &= 0 \\  
  \frac{1}{7} \displaystyle \sum_{j=1}^7 \lambda_j &= 1 \\
  \tau_{j1} = 1.5\ &\text{for }j=1, \ldots, 7 \\
  \tau_{j3} = 3.5\ &\text{for }j=1, \ldots, 7.
\end{align*}
After model estimation, we obtained MAP estimates of the latent variable for each respondent.

\begin{table}[ht]
\centering
\begingroup\footnotesize
\begin{tabular}{lccccccc}
  \hline
 & comfort & environment & work & future & technology & industry & benefit \\ 
  \hline
Trad est & 0.60 & 0.48 & 0.33 & 0.54 & 0.50 & 0.68 & 0.46 \\ 
  \hspace{.1in} (SE) & (0.09) & (0.07) & (0.07) & (0.07) & (0.07) & (0.09) & (0.07) \\ 
  Int est & 0.89 & 1.17 & 0.66 & 0.98 & 1.07 & 1.34 & 0.88 \\ 
  \hspace{.1in} (SE)  & (0.10) & (0.14) & (0.12) & (0.11) & (0.13) & (0.15) & (0.12) \\ 
   \hline
\end{tabular}
\endgroup
\caption{Comparison of loading estimates and SEs under traditional constraints and under integer constraints.} 
\label{tab:loadcomp}
\end{table}

\begin{table}[ht]
\centering
\begingroup\scriptsize
\begin{tabular}{lccccccc}
  \hline
 & comfort & environment & work & future & technology & industry & benefit \\ 
  \hline
Trad est (SE) & $ -2.61 $  (0.21) & $ -1.60 $  (0.11) & $ -1.45 $  (0.10) & $ -2.05 $  (0.14) & $ -1.88 $  (0.12) & $ -2.36 $  (0.17) & $ -1.78 $  (0.12) \\ 
   & $ -1.54 $  (0.11) & $ -0.57 $  (0.07) & $ -0.45 $  (0.07) & $ -0.88 $  (0.08) & $ -0.66 $  (0.08) & $ -1.28 $  (0.10) & $ -0.55 $  (0.07) \\ 
    & $ 0.87 $  (0.09) & $ 0.50 $  (0.07) & $ 1.14 $  (0.08) & $ 0.79 $  (0.08) & $ 0.52 $  (0.07) & $ 0.27 $  (0.08) & $ 0.93 $  (0.08) \\ 
     &  &  &  &  &  &  &  \\ 
  Int est (SE) & $ 1.50 $  (--) & $ 1.50 $  (--) & $ 1.50 $  (--) & $ 1.50 $  (--) & $ 1.50 $  (--) & $ 1.50 $  (--) & $ 1.50 $  (--) \\ 
       & $ 2.12 $  (0.08) & $ 2.48 $  (0.06) & $ 2.27 $  (0.05) & $ 2.32 $  (0.06) & $ 2.52 $  (0.06) & $ 2.32 $  (0.08) & $ 2.41 $  (0.06) \\ 
        & $ 3.50 $  (--) & $ 3.50 $  (--) & $ 3.50 $  (--) & $ 3.50 $  (--) & $ 3.50 $  (--) & $ 3.50 $  (--) & $ 3.50 $  (--) \\ 
   \hline
\end{tabular}
\endgroup
\caption{Comparison of threshold estimates and SEs under traditional constraints and under integer constraints.} 
\label{tab:threshcomp}
\end{table}

\begin{figure}
  \caption{Average of observed variables versus MAP latent variable predictions for the attitudes toward science dataset.}
  \label{fig:avglv}
\begin{knitrout}\footnotesize
\definecolor{shadecolor}{rgb}{0.969, 0.969, 0.969}\color{fgcolor}

{\centering \includegraphics[width=4in,height=4in]{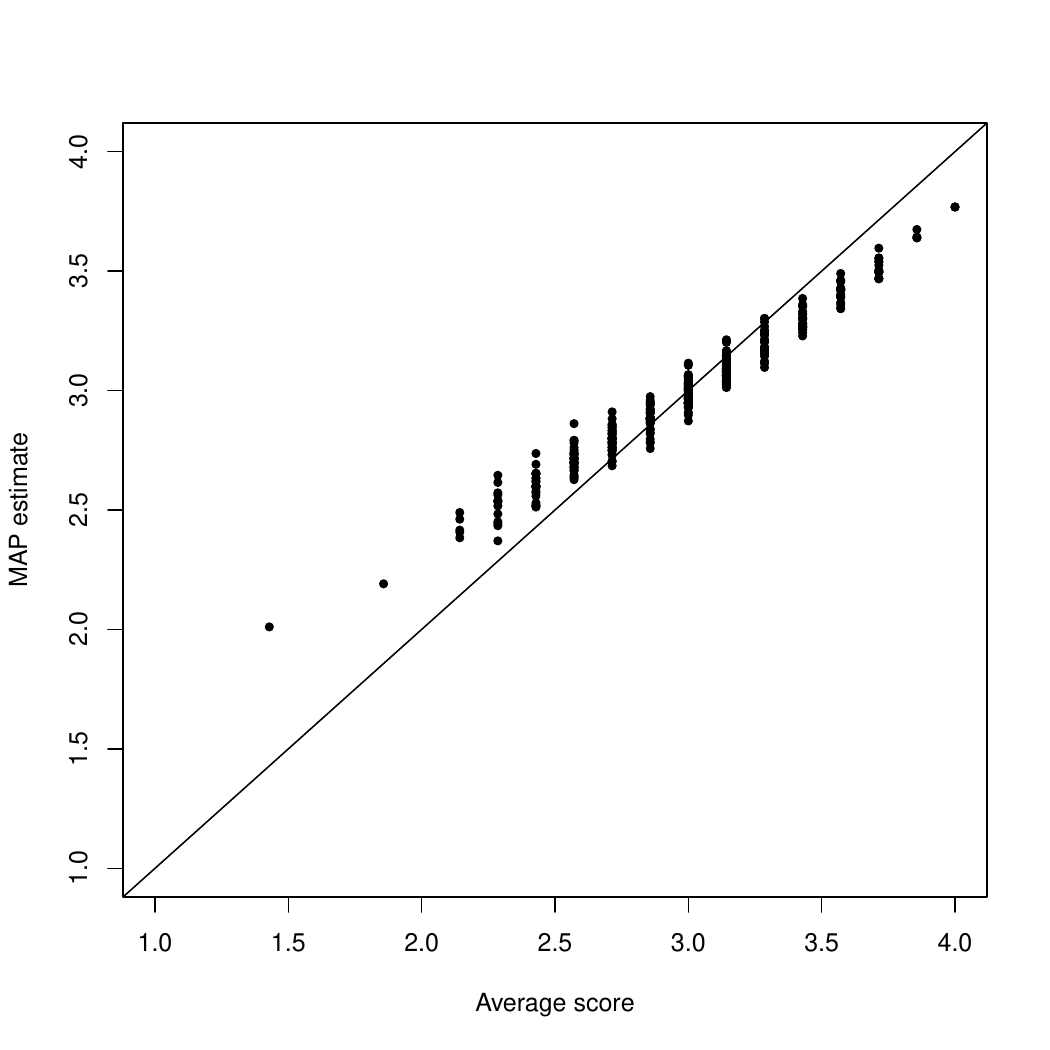} 

}

\end{knitrout}
\end{figure}

\subsection{Results}
As expected, the discrepancy function and $\chi^2$ statistic were identical for the estimated model with traditional identification constraints as compared to the estimated model with the alternative identification constraints. The models do not fit well by any of the traditional fit metrics (e.g., $\chi^2_{14} = 322, p < .01$; RMSEA = 0.24), and poor model fit as well as model misspecifications can lead to questionable parameter interpretations. But because fit is held constant across identification constraints, we proceed with comparing parameter estimates across the two sets of identification constraints.

We begin by comparing estimates of parameters that are shared across the two models. Table~\ref{tab:loadcomp} compares estimated loadings and standard errors under the traditional and alternative constraints, while Table~\ref{tab:threshcomp} does the same for thresholds. Examining Table~\ref{tab:loadcomp}, we see that the loadings and standard errors are larger under the alternative constraints because they are constrained to average 1. The alternative constraints provide a basis for interpreting loadings: values above 1 are larger than average, and values below 1 are smaller than average. The ``work'' item stands out as having the smallest loading under both sets of constraints.

Examining Table~\ref{tab:threshcomp}, many thresholds have no standard errors under the alternative constraints because they are fixed. The free thresholds have standard errors from .05 to .08, which are similar to the standard errors under traditional constraints. Additionally, the threshold estimates under the alternative constraints are intuitive because they can be compared to the 1.5--2.5--3.5 values that would help us to treat the observed variables as continuous. We see that the ``environment'' and ``technology'' items most closely correspond to this pattern, while the middle thresholds for ``comfort'' and ``work'' are noticeably smaller than 2.5. These thresholds interact with the estimated latent mean and variance, which we can freely estimate under the alternative constraints. The estimates are 3 and 0.15, respectively, suggesting that participants generally have high values of the latent variable (attitude toward science). Said differently, the midpoint of a 1--5 scale is 2.5, and the estimated mean of the latent variable is a half-point larger than this midpoint. This result corresponds to the observed response frequencies from Table~\ref{tab:datfreq}.

Finally, Figure~\ref{fig:avglv} compares the average of each participant's ordinal variables to the MAP prediction of the latent variable under the alternative constraints. We see that the MAP predictions are similar to the averages, with some shrinkage whereby the extreme averages have less-extreme latent variable predictions. We also see that the averages and latent variable predictions differ the most for participants with low averages (below 2), reflecting the result that participants generally tended to respond with ``agree'' or ``strongly agree'' on the ordinal scale.

\section{Example 2: Item Response Application}
To further illustrate how the integer constraints work in practice, we now consider a model estimated in an item response framework. We fit our model via marginal maximum likelihood, capitalizing on the flexibility of the {\em mirt} package \cite{cha12} to implement our constraints and to fit the model. In the language of IRT, we can say we are estimating a graded response model with a probit link function. 

\subsection{Method}
We use data from a study of social media privacy \cite{diemet16}, where respondents completed scales related to their use of Facebook and their privacy concerns. We focus on a 5-item subscale of respondents' perceived Facebook benefits that includes items such as ``Facebook allows me to express my personality and feelings.'' Each item contained 5 response categories from ``strongly disagree'' to ``strongly agree.'' The data are available at \url{https://osf.io/e3j98/} and contain responses from 1,156 online participants, where the sampling scheme was designed to be representative of American adults \cite<see>{diemet16}. We model 1,057 participants who supplied complete data on the Facebook benefits scale, which allows for simpler model computations and summaries.

We fit the graded response model with integer constraints in {\em mirt}, making use of package functionality to define new item types and to implement parameter constraints. 
The {\em mirt} marginal maximum likelihood estimation algorithm involves rectangular quadrature with 61 nodes.
The specific integer constraints for this example are:
\begin{align*}
  \displaystyle \sum_{j=1}^5 \nu_j &= 0 \\  
  \frac{1}{5} \displaystyle \sum_{j=1}^5 \lambda_j &= 1 \\
  \tau_{j1} = 1.5\ &\text{for }j=1,\ldots,5 \\
  \tau_{j4} = 4.5\ &\text{for }j=1,\ldots,5.
\end{align*}
To estimate the model with sum constraints on the intercepts and loadings, {\em mirt} makes use of the optimizer from the package {\em Rsolnp} \cite{rsolnp}. This includes a Lagrange multiplier method that can handle both linear and nonlinear parameter constraints.

\subsection{Results}
We first examine model fit, using {\em mirt} to obtain the C2 statistic of \citeA{caimon14}. This statistic rejects the hypothesis of exact fit (C2(df = 3) = 10.03, p = 0.02), which commonly happens in practice. The 90\% confidence interval for RMSEA is (0.017, 0.081), providing some evidence that the model fit is adequate \cite<e.g.,>{may13,mayjoe14}.

\begin{table}[ht]
\centering
\begin{tabular}{lccccc}
  \hline
 & Tau2 & Tau3 & Lambda & Nu & Theta \\ 
  \hline
Item 1 & 2.07 & 3.05 & 0.99 & 0.20 & 0.34 \\ 
  Item 2 & 2.23 & 3.27 & 0.94 & -0.00 & 0.31 \\ 
  Item 3 & 2.16 & 3.11 & 1.04 & -0.05 & 0.27 \\ 
  Item 4 & 2.26 & 3.24 & 0.99 & 0.09 & 0.29 \\ 
  Item 5 & 2.29 & 3.31 & 1.04 & -0.24 & 0.24 \\ 
   \hline
\end{tabular}
\caption{Item parameter estimates for Example 2.} 
\label{tab:ex2est}
\end{table}

Item parameter estimates are shown in Table~\ref{tab:ex2est}. The first two columns are the two free thresholds, followed by the loadings, intercepts, and residual variances (the Tau1 and Tau4 parameters are fixed to 1.5 and 4.5, respectively, for all items). In addition to these parameters, the latent variable mean and variance are estimated to be 2.9
and 0.53, respectively.

From the table, we see that the estimated thresholds for each item are lower than the benchmark values of 2.5 and 3.5. Combined with the fact that the latent variable mean is near the midpoint of 3, this suggests that participants avoided the ``strongly disagree'' option of the scale. The estimated loadings are all near the benchmark value of 1, and no items stand out as being exceptionally better or worse than the others.

\begin{figure}
  \caption{Average of observed variables versus MAP latent variable predictions for the social media dataset.}
  \label{fig:avglv2}
\begin{knitrout}\footnotesize
\definecolor{shadecolor}{rgb}{0.969, 0.969, 0.969}\color{fgcolor}

{\centering \includegraphics[width=4in,height=4in]{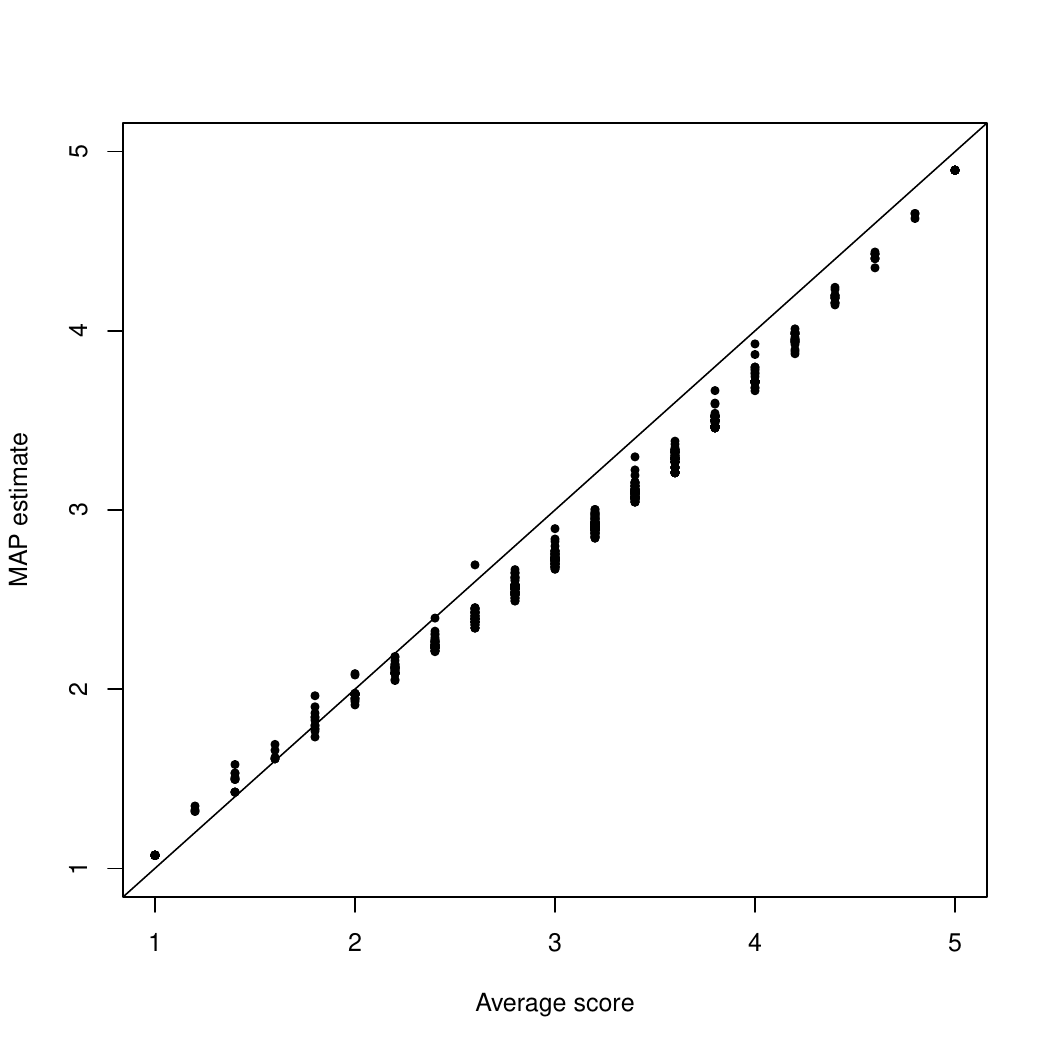} 

}

\end{knitrout}
\end{figure}

Figure~\ref{fig:avglv2} is similar to Figure~\ref{fig:avglv} from our previous example, showing the average score for each individual versus the MAP predictions from the integer-constrained model. We see close agreement here, with points falling slightly below the diagonal for average scores near 3 and larger. This is related to our observation that the thresholds for all items are below the benchmarks of 2.5 and 3.5: the model estimates that people tend to avoid the ``strongly disagree'' option, so lower values of the latent variable can still lead participants to select higher response options.

\begin{figure}
  \caption{Latent variable values versus expected average score (left panel), with overlaid points of MAP latent variable estimates versus observed average scores (right panel).}
  \label{fig:expplot}
\begin{knitrout}\footnotesize
\definecolor{shadecolor}{rgb}{0.969, 0.969, 0.969}\color{fgcolor}

{\centering \includegraphics[width=6.5in]{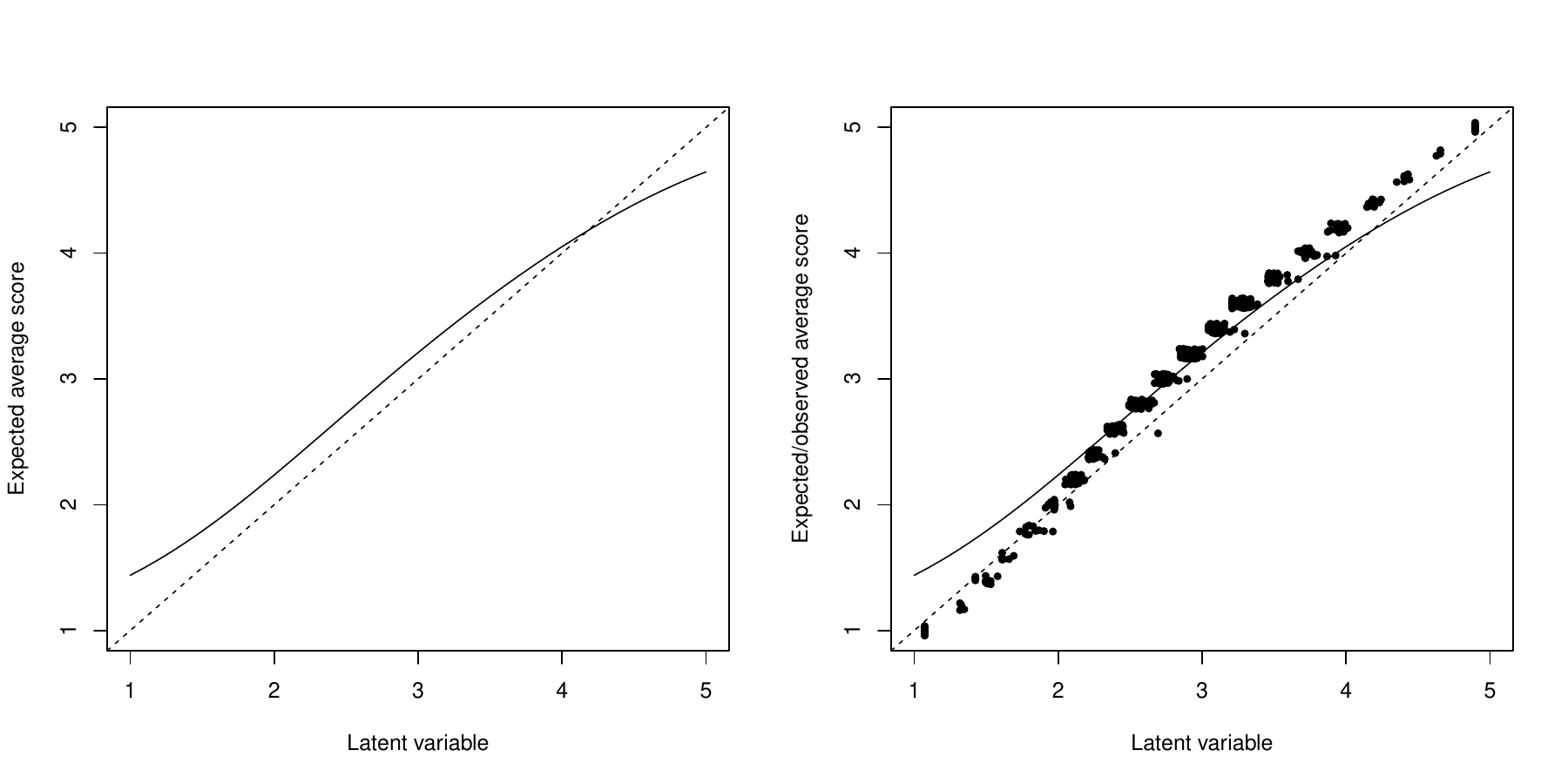} 

}

\end{knitrout}
\end{figure}

A common IRT model summary involves visualization of how the expected test score changes with the latent variable. For ordinal variables, the expected score is typically the sum of integer-coded responses. We consider a similar summary here, showing how the expected average score (the expected test score divided by number of items) varies as a function of the latent variable. In the left panel of Figure~\ref{fig:expplot}, the solid line shows the expected average score (y-axis) for varying values of the latent variable (x-axis) under integer constraints. We see that the line is above the dashed diagonal on the left side of the panel, which is related to the idea that people avoided the ``strongly disagree'' option. That is, participants with low values of the latent variable are expected to have averages greater than 1. The solid line closely follows the diagonal for the rest of the figure, with a crossing near the maximal expected average of 5.

The right panel of Figure~\ref{fig:expplot} contains the same expected average line, with points showing each person's estimated latent variable (x-axis) versus their observed average (y-axis). We see that the points are similar to the expected average near the middle (latent variables of 2.5 to 3.5), and they stray from the expected average near the extremes. This difference is likely due to the boundaries of the expected average score. That is, the average score has a hard lower bound of 1 and a hard upper bound of 5, so the expected value will be pulled towards the center of the scale. The latent variables are unbounded, allowing us to observe predictions near the extremes of 1 and 5.

The {\em mirt} model estimation that we implemented here is not comprehensive. Most notably, we did not handle missing values, and we did not obtain standard errors of parameter estimates under integer constraints. The latter requires further analytical work on the Hessian under integer constraints, or a Jacobian so that we can apply the delta method to standard errors under traditional constraints. But we have illustrated that the constraints can be applied in traditional IRT settings, where latent variables are commonly used for scoring purposes.

\section{General Discussion}
In this paper, we first considered how constraints on an item factor analysis model can lead to latent variable predictions mimicking the average of observed ordinal variables, where the variables are coded as 1, 2, \ldots, $K$. Based on these constraints, we then defined a set of minimal identification constraints (``integer constraints'') that puts the latent variable on the scale of the integer-coded ordinal variable. This is potentially worthwhile because applied researchers are accustomed to thinking on the scale of the ordinal variable and to treating ordinal variables as though they are continuous. Our simulation showed that the integer constraints did not meaningfully influence rates of model convergence or admissibility, at least for the conditions examined. Our examples showed specific uses of the constraints, including enhanced interpretation of parameter estimates, intuitive latent variable predictions, and application in traditional IRT settings. In the sections below, we consider additional uses of integer constraints and potential extensions of our results.

\subsection{Additional Applications}
The integer constraints suggest a likelihood ratio test of whether or not the observed ordinal variables can be treated as continuous. That is, we can fit an ordinal CFA with our proposed minimal identification constraints, then conduct a likelihood ratio test comparing this model to the highly-constrained model whose latent variable predictions are the observed averages. If the likelihood ratio test suggests that the fit of the two models is equal, then researchers could feel more confident about treating their ordinal variables as continuous. We are doubtful that this likelihood ratio test will often indicate that the fit of the two models is equal.

Related to the likelihood ratio test, integer constraints could be further considered in the context of measurement invariance studies with multiple groups. \citeA{wuest16} provide a comprehensive treatment of measurement invariance under traditional identification constraints, and we used some of their results in this paper. Because we have defined transformation matrices to convert traditional constraints to integer constraints, much of the Wu and Estabrook results could be translated to testing measurement invariance under integer constraints. The integer constraints may help to make measurement invariance testing more interpretable and intuitive.

Finally, the integer constraints have potential uses in Bayesian modeling because they potentially make specification of prior distributions more intuitive. For example, because loadings are constrained to average 1 under integer constraints, the priors for factor loadings would often have a mean of 1. And because the factor mean is related to the average of ordinal variables, researchers may more easily convert their prior expectations to prior distributions. On the other hand, the sum constraints on loadings can complicate the prior distributions of those parameters \cite<e.g.,>{merari23}, so that Bayesian SEM software may not automatically handle the constraints. One possible solution involves discarding the sum constraints on loadings, replacing them with constraints on a single loading per factor.

If we are to maintain the sum constraints in a Bayesian context, another possible solution involves estimating the model using the traditional identification constraints (which are available in most software) while specifying priors for integer-constrained parameters. Then we would additionally need the Jacobian for transforming the parameters under traditional constraints to the parameters under integer constraints, which involves the results from Remark 2. For many models, this Jacobian will involve the determinant of a large matrix, though the structure of the underlying matrix may allow for fast determinant computations. Further work is needed.

\subsection{Differing Numbers of Categories per Variable}
The developments in this paper relied on the assumption that the observed variables are all ordinal with $K$ categories. In practice, it is common to have ordinal variables with differing numbers of categories, for example two ordinal variables with three possible categories and three ordinal variables with five possible categories. In this case, we write that each ordinal variable $j$ has $K_j$ response categories. It is more cumbersome to specify identification constraints here, because the differing number of thresholds per variable complicates matrix manipulations such as Equation~\eqref{eq:tran1}.

Because integer constraints were designed to be close to the average of integer-coded variables, we should also consider whether it makes sense to take an average when variables have different numbers of categories. As an extreme example, consider a situation where two variables have three categories and a third variable has 50 categories. If we code each variable using integers starting at 1 and then average them, it is clear that the third variable will usually dominate the average. This suggests that we should view the 3-category variables as coarsened versions of the 50-category scale. That is, if we were to code the 3-category variables so that they assumed values on a scale from 1 to 50, then it would make more sense to average across the variables. In the context of our integer constraints, this amounts to fixing the thresholds of the 3-category variables to values other than 1.5 and 2.5. The thresholds should instead divide the 50-point scale into three equal segments, which here corresponds to a lower threshold of 17.17 and an upper threshold of 33.83.

In general terms, let $K_j$ be the number of categories for ordinal variable $j$. For a specific latent variable $k$, let $K_{\text{max}} = \text{max}_{j \in \mathcal{S}_k} K_j$. Then for all observed variables that load on latent variable $k$ (i.e., for all $j \in \mathcal{S}_k$), we should fix
\begin{align*}
  \tau_{j1} &= \frac{1}{2} + \frac{K_{\text{max}}}{K_j} \\
  \tau_{j(K_j-1)} &= \frac{1}{2} + \frac{K_{\text{max}}(K_j - 1)}{K_j}.
\end{align*}
The other constraints on intercepts and loadings remain the same as before. A modification is additionally required for a binary variable $j$, which only has a single threshold parameter. For that case, we fix $\tau_{j1}$ in the above manner while also fixing the intercept $\nu_j$ to 0.

\subsection{Summary}
The sum score is a major consideration in the historical development of psychometrics as well as in current developments \cite<e.g.,>{mcn24,mis24,sijell24}. In this paper, we studied an integer identification constraint for ordinal CFA that has a direct relationship to the sum score, where we average the ordinal items as if they are continuous. These constraints might balance the concerns of those who view the sum score as unsophisticated with those who view the sum score as a benchmark. In drawing on the intuition of the sum score, we hope that the integer constraints will enable more researchers to meaningfully employ common psychometric models of ordinal data. 

\section*{Computational Details}

All results were obtained using the \proglang{R}~system for statistical computing \cite{r23},
version~4.5.2, making use of the {\em lavaan} \cite{ros12}, {\em ggplot2} \cite{ggplot2}, and {\em xtable} \cite{xtable} packages. Code for reproducing the results in this paper and for applying integer constraints to other {\em lavaan} models is available at \url{https://semtools.r-forge.r-project.org/}.

\bibliography{refs}

\begin{thebibliography}{}

\bibitem [\protect \citeauthoryear {%
Andersen%
}{%
Andersen%
}{%
{\protect \APACyear {1977}}%
}]{%
and77}
\APACinsertmetastar {%
and77}%
\begin{APACrefauthors}%
Andersen, E\BPBI B.%
\end{APACrefauthors}%
\unskip\
\newblock
\APACrefYearMonthDay{1977}{}{}.
\newblock
{\BBOQ}\APACrefatitle {Sufficient statistics and latent trait models}
  {Sufficient statistics and latent trait models}.{\BBCQ}
\newblock
\APACjournalVolNumPages{Psychometrika}{42}{1}{69--81}.
\newblock
\begin{APACrefDOI} \doi{10.1007/bf02293746} \end{APACrefDOI}
\PrintBackRefs{\CurrentBib}

\bibitem [\protect \citeauthoryear {%
Andrich%
}{%
Andrich%
}{%
{\protect \APACyear {1978}}%
}]{%
and78}
\APACinsertmetastar {%
and78}%
\begin{APACrefauthors}%
Andrich, D.%
\end{APACrefauthors}%
\unskip\
\newblock
\APACrefYearMonthDay{1978}{}{}.
\newblock
{\BBOQ}\APACrefatitle {Application of a psychometric rating model to ordered
  categories which are scored with successive integers} {Application of a
  psychometric rating model to ordered categories which are scored with
  successive integers}.{\BBCQ}
\newblock
\APACjournalVolNumPages{Applied Psychological Measurement}{2}{4}{581--594}.
\PrintBackRefs{\CurrentBib}

\bibitem [\protect \citeauthoryear {%
Baker%
\ \BBA {} Kim%
}{%
Baker%
\ \BBA {} Kim%
}{%
{\protect \APACyear {2004}}%
}]{%
bakkim04}
\APACinsertmetastar {%
bakkim04}%
\begin{APACrefauthors}%
Baker, F\BPBI B.%
\BCBT {}\ \BBA {} Kim, S\BHBI H.%
\end{APACrefauthors}%
\unskip\
\newblock
\APACrefYear{2004}.
\newblock
\APACrefbtitle {Item response theory: {P}arameter estimation techniques} {Item
  response theory: {P}arameter estimation techniques}\ (\PrintOrdinal{2nd}\
  \BEd).
\newblock
\APACaddressPublisher{}{Boca Raton, FL:\ CRC Press}.
\PrintBackRefs{\CurrentBib}

\bibitem [\protect \citeauthoryear {%
Ben-Shachar%
, Lüdecke%
\BCBL {}\ \BBA {} Makowski%
}{%
Ben-Shachar%
\ \protect \BOthers {.}}{%
{\protect \APACyear {2020}}%
}]{%
benshachar}
\APACinsertmetastar {%
benshachar}%
\begin{APACrefauthors}%
Ben-Shachar, M\BPBI S.%
, Lüdecke, D.%
\BCBL {}\ \BBA {} Makowski, D.%
\end{APACrefauthors}%
\unskip\
\newblock
\APACrefYearMonthDay{2020}{}{}.
\newblock
{\BBOQ}\APACrefatitle {{e}ffectsize: Estimation of Effect Size Indices and
  Standardized Parameters} {{e}ffectsize: Estimation of effect size indices and
  standardized parameters}.{\BBCQ}
\newblock
\APACjournalVolNumPages{Journal of Open Source Software}{5}{56}{2815}.
\newblock
\begin{APACrefURL} \url{https://doi.org/10.21105/joss.02815} \end{APACrefURL}
\newblock
\begin{APACrefDOI} \doi{10.21105/joss.02815} \end{APACrefDOI}
\PrintBackRefs{\CurrentBib}

\bibitem [\protect \citeauthoryear {%
Bollen%
\ \BBA {} Barb%
}{%
Bollen%
\ \BBA {} Barb%
}{%
{\protect \APACyear {1981}}%
}]{%
bolbar81}
\APACinsertmetastar {%
bolbar81}%
\begin{APACrefauthors}%
Bollen, K\BPBI A.%
\BCBT {}\ \BBA {} Barb, K\BPBI H.%
\end{APACrefauthors}%
\unskip\
\newblock
\APACrefYearMonthDay{1981}{}{}.
\newblock
{\BBOQ}\APACrefatitle {Pearson's {R} and coarsely categorized measures}
  {Pearson's {R} and coarsely categorized measures}.{\BBCQ}
\newblock
\APACjournalVolNumPages{American Sociological Review}{46}{}{232--239}.
\PrintBackRefs{\CurrentBib}

\bibitem [\protect \citeauthoryear {%
Bollen%
, Lilly%
\BCBL {}\ \BBA {} Luo%
}{%
Bollen%
\ \protect \BOthers {.}}{%
{\protect \APACyear {2024}}%
}]{%
bollil24}
\APACinsertmetastar {%
bollil24}%
\begin{APACrefauthors}%
Bollen, K\BPBI A.%
, Lilly, A\BPBI G.%
\BCBL {}\ \BBA {} Luo, L.%
\end{APACrefauthors}%
\unskip\
\newblock
\APACrefYearMonthDay{2024}{}{}.
\newblock
{\BBOQ}\APACrefatitle {Selecting scaling indicators in structural equation
  models {(SEMs)}} {Selecting scaling indicators in structural equation models
  {(SEMs)}}.{\BBCQ}
\newblock
\APACjournalVolNumPages{Psychological Methods}{29}{}{868--889}.
\PrintBackRefs{\CurrentBib}

\bibitem [\protect \citeauthoryear {%
B\"{u}rkner%
\ \BBA {} Vuorre%
}{%
B\"{u}rkner%
\ \BBA {} Vuorre%
}{%
{\protect \APACyear {2019}}%
}]{%
burvuo19}
\APACinsertmetastar {%
burvuo19}%
\begin{APACrefauthors}%
B\"{u}rkner, P\BHBI C.%
\BCBT {}\ \BBA {} Vuorre, M.%
\end{APACrefauthors}%
\unskip\
\newblock
\APACrefYearMonthDay{2019}{}{}.
\newblock
{\BBOQ}\APACrefatitle {Ordinal Regression Models in Psychology: {A} Tutorial}
  {Ordinal regression models in psychology: {A} tutorial}.{\BBCQ}
\newblock
\APACjournalVolNumPages{Advances in Methods and Practices in Psychological
  Science}{2}{1}{77--101}.
\newblock
\begin{APACrefDOI} \doi{10.1177/2515245918823199} \end{APACrefDOI}
\PrintBackRefs{\CurrentBib}

\bibitem [\protect \citeauthoryear {%
Cai%
\ \BBA {} Monroe%
}{%
Cai%
\ \BBA {} Monroe%
}{%
{\protect \APACyear {2014}}%
}]{%
caimon14}
\APACinsertmetastar {%
caimon14}%
\begin{APACrefauthors}%
Cai, L.%
\BCBT {}\ \BBA {} Monroe, S.%
\end{APACrefauthors}%
\unskip\
\newblock
\APACrefYearMonthDay{2014}{}{}.
\newblock
{\BBOQ}\APACrefatitle {A New Statistic for Evaluating Item Response Theory
  Models for Ordinal Data} {A new statistic for evaluating item response theory
  models for ordinal data}.{\BBCQ}
\newblock
\APACjournalVolNumPages{{CRESST} Report 839}{}{}{}.
\PrintBackRefs{\CurrentBib}

\bibitem [\protect \citeauthoryear {%
Chalmers%
}{%
Chalmers%
}{%
{\protect \APACyear {2012}}%
}]{%
cha12}
\APACinsertmetastar {%
cha12}%
\begin{APACrefauthors}%
Chalmers, R\BPBI P.%
\end{APACrefauthors}%
\unskip\
\newblock
\APACrefYearMonthDay{2012}{}{}.
\newblock
{\BBOQ}\APACrefatitle {{mirt}: {A} Multidimensional Item Response Theory
  Package for the {R} Environment} {{mirt}: {A} multidimensional item response
  theory package for the {R} environment}.{\BBCQ}
\newblock
\APACjournalVolNumPages{Journal of Statistical Software}{48}{6}{1--29}.
\newblock
\begin{APACrefURL} \url{http://www.jstatsoft.org/v48/i06/} \end{APACrefURL}
\PrintBackRefs{\CurrentBib}

\bibitem [\protect \citeauthoryear {%
Dahl%
, Scott%
, Roosen%
, Magnusson%
\BCBL {}\ \BBA {} Swinton%
}{%
Dahl%
\ \protect \BOthers {.}}{%
{\protect \APACyear {2019}}%
}]{%
xtable}
\APACinsertmetastar {%
xtable}%
\begin{APACrefauthors}%
Dahl, D\BPBI B.%
, Scott, D.%
, Roosen, C.%
, Magnusson, A.%
\BCBL {}\ \BBA {} Swinton, J.%
\end{APACrefauthors}%
\unskip\
\newblock
\APACrefYearMonthDay{2019}{}{}.
\newblock
{\BBOQ}\APACrefatitle {{xtable: Export} Tables to LaTeX or HTML} {{xtable:
  Export} tables to latex or html}{\BBCQ}\ [\bibcomputersoftwaremanual].
\newblock
\begin{APACrefURL} \url{https://CRAN.R-project.org/package=xtable}
  \end{APACrefURL}
\newblock
\APACrefnote{R package version 1.8-4}
\PrintBackRefs{\CurrentBib}

\bibitem [\protect \citeauthoryear {%
Dienlin%
\ \BBA {} Metzger%
}{%
Dienlin%
\ \BBA {} Metzger%
}{%
{\protect \APACyear {2016}}%
}]{%
diemet16}
\APACinsertmetastar {%
diemet16}%
\begin{APACrefauthors}%
Dienlin, T.%
\BCBT {}\ \BBA {} Metzger, M\BPBI J.%
\end{APACrefauthors}%
\unskip\
\newblock
\APACrefYearMonthDay{2016}{}{}.
\newblock
{\BBOQ}\APACrefatitle {An extended privacy calculus model for {SNSs: Analyzing}
  self-disclosure and self-withdrawal in a representative US sample} {An
  extended privacy calculus model for {SNSs: Analyzing} self-disclosure and
  self-withdrawal in a representative us sample}.{\BBCQ}
\newblock
\APACjournalVolNumPages{Journal of Computer-Mediated
  Communication}{21}{5}{368--383}.
\PrintBackRefs{\CurrentBib}

\bibitem [\protect \citeauthoryear {%
Flora%
\ \BBA {} Curran%
}{%
Flora%
\ \BBA {} Curran%
}{%
{\protect \APACyear {2004}}%
}]{%
flora2004empirical}
\APACinsertmetastar {%
flora2004empirical}%
\begin{APACrefauthors}%
Flora, D\BPBI B.%
\BCBT {}\ \BBA {} Curran, P\BPBI J.%
\end{APACrefauthors}%
\unskip\
\newblock
\APACrefYearMonthDay{2004}{}{}.
\newblock
{\BBOQ}\APACrefatitle {An empirical evaluation of alternative methods of
  estimation for confirmatory factor analysis with ordinal data.} {An empirical
  evaluation of alternative methods of estimation for confirmatory factor
  analysis with ordinal data.}{\BBCQ}
\newblock
\APACjournalVolNumPages{Psychological Methods}{9}{4}{466}.
\PrintBackRefs{\CurrentBib}

\bibitem [\protect \citeauthoryear {%
Foldnes%
\ \BBA {} Grønneberg%
}{%
Foldnes%
\ \BBA {} Grønneberg%
}{%
{\protect \APACyear {2022}}%
}]{%
folgro22}
\APACinsertmetastar {%
folgro22}%
\begin{APACrefauthors}%
Foldnes, N.%
\BCBT {}\ \BBA {} Grønneberg, S.%
\end{APACrefauthors}%
\unskip\
\newblock
\APACrefYearMonthDay{2022}{}{}.
\newblock
{\BBOQ}\APACrefatitle {The sensitivity of structural equation modeling with
  ordinal data to underlying non-normality and observed distributional forms}
  {The sensitivity of structural equation modeling with ordinal data to
  underlying non-normality and observed distributional forms}.{\BBCQ}
\newblock
\APACjournalVolNumPages{Psychological Methods}{27}{4}{541–567}.
\newblock
\begin{APACrefDOI} \doi{10.1037/met0000385} \end{APACrefDOI}
\PrintBackRefs{\CurrentBib}

\bibitem [\protect \citeauthoryear {%
Gagn{\'e}%
\ \BBA {} Hancock%
}{%
Gagn{\'e}%
\ \BBA {} Hancock%
}{%
{\protect \APACyear {2006}}%
}]{%
gagne2006measurement}
\APACinsertmetastar {%
gagne2006measurement}%
\begin{APACrefauthors}%
Gagn{\'e}, P.%
\BCBT {}\ \BBA {} Hancock, G\BPBI R.%
\end{APACrefauthors}%
\unskip\
\newblock
\APACrefYearMonthDay{2006}{}{}.
\newblock
{\BBOQ}\APACrefatitle {Measurement model quality, sample size, and solution
  propriety in confirmatory factor models} {Measurement model quality, sample
  size, and solution propriety in confirmatory factor models}.{\BBCQ}
\newblock
\APACjournalVolNumPages{Multivariate Behavioral Research}{41}{1}{65--83}.
\PrintBackRefs{\CurrentBib}

\bibitem [\protect \citeauthoryear {%
Galanos%
\ \BBA {} Ye%
}{%
Galanos%
\ \BBA {} Ye%
}{%
{\protect \APACyear {2025}}%
}]{%
rsolnp}
\APACinsertmetastar {%
rsolnp}%
\begin{APACrefauthors}%
Galanos, A.%
\BCBT {}\ \BBA {} Ye, Y.%
\end{APACrefauthors}%
\unskip\
\newblock
\APACrefYearMonthDay{2025}{}{}.
\newblock
{\BBOQ}\APACrefatitle {{\em Rsolnp}: {General} Non-Linear Optimization} {{\em
  Rsolnp}: {General} non-linear optimization}{\BBCQ}\
  [\bibcomputersoftwaremanual].
\newblock
\APACrefnote{R package version 2.0.1}
\PrintBackRefs{\CurrentBib}

\bibitem [\protect \citeauthoryear {%
Graves%
\ \BBA {} Merkle%
}{%
Graves%
\ \BBA {} Merkle%
}{%
{\protect \APACyear {2022}}%
}]{%
gramer22}
\APACinsertmetastar {%
gramer22}%
\begin{APACrefauthors}%
Graves, B.%
\BCBT {}\ \BBA {} Merkle, E\BPBI C.%
\end{APACrefauthors}%
\unskip\
\newblock
\APACrefYearMonthDay{2022}{}{}.
\newblock
{\BBOQ}\APACrefatitle {A note on identification constraints and information
  criteria in {Bayesian} latent variable models} {A note on identification
  constraints and information criteria in {Bayesian} latent variable
  models}.{\BBCQ}
\newblock
\APACjournalVolNumPages{Behavior Research Methods}{54}{}{795--804}.
\PrintBackRefs{\CurrentBib}

\bibitem [\protect \citeauthoryear {%
Grønneberg%
\ \BBA {} Foldnes%
}{%
Grønneberg%
\ \BBA {} Foldnes%
}{%
{\protect \APACyear {2024}}%
}]{%
grofol24}
\APACinsertmetastar {%
grofol24}%
\begin{APACrefauthors}%
Grønneberg, S.%
\BCBT {}\ \BBA {} Foldnes, N.%
\end{APACrefauthors}%
\unskip\
\newblock
\APACrefYearMonthDay{2024}{}{}.
\newblock
{\BBOQ}\APACrefatitle {Factor analyzing ordinal items requires substantive
  knowledge of response marginals} {Factor analyzing ordinal items requires
  substantive knowledge of response marginals}.{\BBCQ}
\newblock
\APACjournalVolNumPages{Psychological Methods}{29}{1}{65–87}.
\newblock
\begin{APACrefDOI} \doi{10.1037/met0000495} \end{APACrefDOI}
\PrintBackRefs{\CurrentBib}

\bibitem [\protect \citeauthoryear {%
Klopp%
\ \BBA {} Kl{\"o}{\ss}ner%
}{%
Klopp%
\ \BBA {} Kl{\"o}{\ss}ner%
}{%
{\protect \APACyear {2023}}%
}]{%
klo23}
\APACinsertmetastar {%
klo23}%
\begin{APACrefauthors}%
Klopp, E.%
\BCBT {}\ \BBA {} Kl{\"o}{\ss}ner, S.%
\end{APACrefauthors}%
\unskip\
\newblock
\APACrefYearMonthDay{2023}{}{}.
\newblock
{\BBOQ}\APACrefatitle {Scaling Metric Measurement Invariance Models} {Scaling
  metric measurement invariance models}.{\BBCQ}
\newblock
\APACjournalVolNumPages{Methodology}{19}{3}{192--227}.
\PrintBackRefs{\CurrentBib}

\bibitem [\protect \citeauthoryear {%
Klopp%
\ \BBA {} Klößner%
}{%
Klopp%
\ \BBA {} Klößner%
}{%
{\protect \APACyear {2021}}%
}]{%
klo21}
\APACinsertmetastar {%
klo21}%
\begin{APACrefauthors}%
Klopp, E.%
\BCBT {}\ \BBA {} Klößner, S.%
\end{APACrefauthors}%
\unskip\
\newblock
\APACrefYearMonthDay{2021}{}{}.
\newblock
{\BBOQ}\APACrefatitle {The Impact of Scaling Methods on the Properties and
  Interpretation of Parameter Estimates in Structural Equation Models with
  Latent Variables} {The impact of scaling methods on the properties and
  interpretation of parameter estimates in structural equation models with
  latent variables}.{\BBCQ}
\newblock
\APACjournalVolNumPages{Structural Equation Modeling: {A} Multidisciplinary
  Journal}{28}{2}{182–206}.
\newblock
\begin{APACrefDOI} \doi{10.1080/10705511.2020.1796673} \end{APACrefDOI}
\PrintBackRefs{\CurrentBib}

\bibitem [\protect \citeauthoryear {%
Klößner%
\ \BBA {} Klopp%
}{%
Klößner%
\ \BBA {} Klopp%
}{%
{\protect \APACyear {2018}}%
}]{%
kloklo18}
\APACinsertmetastar {%
kloklo18}%
\begin{APACrefauthors}%
Klößner, S.%
\BCBT {}\ \BBA {} Klopp, E.%
\end{APACrefauthors}%
\unskip\
\newblock
\APACrefYearMonthDay{2018}{}{}.
\newblock
{\BBOQ}\APACrefatitle {Explaining Constraint Interaction: {How} to Interpret
  Estimated Model Parameters Under Alternative Scaling Methods} {Explaining
  constraint interaction: {How} to interpret estimated model parameters under
  alternative scaling methods}.{\BBCQ}
\newblock
\APACjournalVolNumPages{Structural Equation Modeling: {A} Multidisciplinary
  Journal}{26}{1}{143–155}.
\newblock
\begin{APACrefDOI} \doi{10.1080/10705511.2018.1517356} \end{APACrefDOI}
\PrintBackRefs{\CurrentBib}

\bibitem [\protect \citeauthoryear {%
Kruschke%
}{%
Kruschke%
}{%
{\protect \APACyear {2014}}%
}]{%
kru14}
\APACinsertmetastar {%
kru14}%
\begin{APACrefauthors}%
Kruschke, J.%
\end{APACrefauthors}%
\unskip\
\newblock
\APACrefYearMonthDay{2014}{}{}.
\newblock
\APACrefbtitle {Ordinal probit regression: {Transforming} polr() parameter
  values to make them more intuitive.} {Ordinal probit regression:
  {Transforming} polr() parameter values to make them more intuitive.}
\newblock
\begin{APACrefURL}
  \url{http://doingbayesiandataanalysis.blogspot.com/2014/11/ordinal-probit-regression-transforming.html}
  \end{APACrefURL}
\PrintBackRefs{\CurrentBib}

\bibitem [\protect \citeauthoryear {%
Kruschke%
}{%
Kruschke%
}{%
{\protect \APACyear {2015}}%
}]{%
kru15}
\APACinsertmetastar {%
kru15}%
\begin{APACrefauthors}%
Kruschke, J.%
\end{APACrefauthors}%
\unskip\
\newblock
\APACrefYear{2015}.
\newblock
\APACrefbtitle {Doing {Bayesian} data analysis: {A} tutorial with {R, JAGS, and
  Stan}} {Doing {Bayesian} data analysis: {A} tutorial with {R, JAGS, and
  Stan}}\ (\PrintOrdinal{2nd}\ \BEd).
\newblock
\APACaddressPublisher{}{Academic Press}.
\PrintBackRefs{\CurrentBib}

\bibitem [\protect \citeauthoryear {%
Lee%
}{%
Lee%
}{%
{\protect \APACyear {2007}}%
}]{%
lee07}
\APACinsertmetastar {%
lee07}%
\begin{APACrefauthors}%
Lee, S\BHBI Y.%
\end{APACrefauthors}%
\unskip\
\newblock
\APACrefYear{2007}.
\newblock
\APACrefbtitle {Structural Equation Modeling: {A} {B}ayesian Approach}
  {Structural equation modeling: {A} {B}ayesian approach}.
\newblock
\APACaddressPublisher{}{Chichester, England:\ John Wiley \& Sons}.
\PrintBackRefs{\CurrentBib}

\bibitem [\protect \citeauthoryear {%
Lee%
, Poon%
\BCBL {}\ \BBA {} Bentler%
}{%
Lee%
\ \protect \BOthers {.}}{%
{\protect \APACyear {1990}}%
}]{%
lee90}
\APACinsertmetastar {%
lee90}%
\begin{APACrefauthors}%
Lee, S\BHBI Y.%
, Poon, W\BHBI Y.%
\BCBL {}\ \BBA {} Bentler, P.%
\end{APACrefauthors}%
\unskip\
\newblock
\APACrefYearMonthDay{1990}{}{}.
\newblock
{\BBOQ}\APACrefatitle {Full maximum likelihood analysis of structural equation
  models with polytomous variables} {Full maximum likelihood analysis of
  structural equation models with polytomous variables}.{\BBCQ}
\newblock
\APACjournalVolNumPages{Statistics \& Probability Letters}{9}{1}{91--97}.
\PrintBackRefs{\CurrentBib}

\bibitem [\protect \citeauthoryear {%
Liddell%
\ \BBA {} Kruschke%
}{%
Liddell%
\ \BBA {} Kruschke%
}{%
{\protect \APACyear {2018}}%
}]{%
lidkru18}
\APACinsertmetastar {%
lidkru18}%
\begin{APACrefauthors}%
Liddell, T\BPBI M.%
\BCBT {}\ \BBA {} Kruschke, J\BPBI K.%
\end{APACrefauthors}%
\unskip\
\newblock
\APACrefYearMonthDay{2018}{}{}.
\newblock
{\BBOQ}\APACrefatitle {Analyzing ordinal data with metric models: What could
  possibly go wrong?} {Analyzing ordinal data with metric models: What could
  possibly go wrong?}{\BBCQ}
\newblock
\APACjournalVolNumPages{Journal of Experimental Social
  Psychology}{79}{}{328--348}.
\newblock
\begin{APACrefDOI} \doi{https://doi.org/10.1016/j.jesp.2018.08.009}
  \end{APACrefDOI}
\PrintBackRefs{\CurrentBib}

\bibitem [\protect \citeauthoryear {%
Little%
, Slegers%
\BCBL {}\ \BBA {} Card%
}{%
Little%
\ \protect \BOthers {.}}{%
{\protect \APACyear {2006}}%
}]{%
little06}
\APACinsertmetastar {%
little06}%
\begin{APACrefauthors}%
Little, T\BPBI D.%
, Slegers, D\BPBI W.%
\BCBL {}\ \BBA {} Card, N\BPBI A.%
\end{APACrefauthors}%
\unskip\
\newblock
\APACrefYearMonthDay{2006}{}{}.
\newblock
{\BBOQ}\APACrefatitle {A non-arbitrary method of identifying and scaling latent
  variables in {SEM} and {MACS} models} {A non-arbitrary method of identifying
  and scaling latent variables in {SEM} and {MACS} models}.{\BBCQ}
\newblock
\APACjournalVolNumPages{Structural Equation Modeling}{13}{1}{59--72}.
\PrintBackRefs{\CurrentBib}

\bibitem [\protect \citeauthoryear {%
Lord%
}{%
Lord%
}{%
{\protect \APACyear {1953}}%
}]{%
lord53}
\APACinsertmetastar {%
lord53}%
\begin{APACrefauthors}%
Lord, F\BPBI M.%
\end{APACrefauthors}%
\unskip\
\newblock
\APACrefYearMonthDay{1953}{}{}.
\newblock
{\BBOQ}\APACrefatitle {An application of confidence intervals and of maximum
  likelihood to the estimation of an examinee's ability} {An application of
  confidence intervals and of maximum likelihood to the estimation of an
  examinee's ability}.{\BBCQ}
\newblock
\APACjournalVolNumPages{Psychometrika}{18}{}{57--76}.
\PrintBackRefs{\CurrentBib}

\bibitem [\protect \citeauthoryear {%
Maydeu-Olivares%
}{%
Maydeu-Olivares%
}{%
{\protect \APACyear {2013}}%
}]{%
may13}
\APACinsertmetastar {%
may13}%
\begin{APACrefauthors}%
Maydeu-Olivares, A.%
\end{APACrefauthors}%
\unskip\
\newblock
\APACrefYearMonthDay{2013}{}{}.
\newblock
{\BBOQ}\APACrefatitle {Goodness-of-fit assessment of item response theory
  models} {Goodness-of-fit assessment of item response theory models}.{\BBCQ}
\newblock
\APACjournalVolNumPages{Measurement: Interdisciplinary Research and
  Perspectives}{11}{3}{71--101}.
\PrintBackRefs{\CurrentBib}

\bibitem [\protect \citeauthoryear {%
Maydeu-Olivares%
\ \BBA {} Joe%
}{%
Maydeu-Olivares%
\ \BBA {} Joe%
}{%
{\protect \APACyear {2014}}%
}]{%
mayjoe14}
\APACinsertmetastar {%
mayjoe14}%
\begin{APACrefauthors}%
Maydeu-Olivares, A.%
\BCBT {}\ \BBA {} Joe, H.%
\end{APACrefauthors}%
\unskip\
\newblock
\APACrefYearMonthDay{2014}{}{}.
\newblock
{\BBOQ}\APACrefatitle {Assessing approximate fit in categorical data analysis}
  {Assessing approximate fit in categorical data analysis}.{\BBCQ}
\newblock
\APACjournalVolNumPages{Multivariate Behavioral Research}{49}{4}{305--328}.
\PrintBackRefs{\CurrentBib}

\bibitem [\protect \citeauthoryear {%
McNeish%
}{%
McNeish%
}{%
{\protect \APACyear {2024}}%
}]{%
mcn24}
\APACinsertmetastar {%
mcn24}%
\begin{APACrefauthors}%
McNeish, D.%
\end{APACrefauthors}%
\unskip\
\newblock
\APACrefYearMonthDay{2024}{}{}.
\newblock
{\BBOQ}\APACrefatitle {Practical Implications of Sum Scores Being
  Psychometrics’ Greatest Accomplishment} {Practical implications of sum
  scores being psychometrics’ greatest accomplishment}.{\BBCQ}
\newblock
\APACjournalVolNumPages{Psychometrika}{89}{4}{1148–1169}.
\newblock
\begin{APACrefDOI} \doi{10.1007/s11336-024-09988-z} \end{APACrefDOI}
\PrintBackRefs{\CurrentBib}

\bibitem [\protect \citeauthoryear {%
McNeish%
\ \BBA {} Wolf%
}{%
McNeish%
\ \BBA {} Wolf%
}{%
{\protect \APACyear {2020}}%
}]{%
mcn20}
\APACinsertmetastar {%
mcn20}%
\begin{APACrefauthors}%
McNeish, D.%
\BCBT {}\ \BBA {} Wolf, M\BPBI G.%
\end{APACrefauthors}%
\unskip\
\newblock
\APACrefYearMonthDay{2020}{}{}.
\newblock
{\BBOQ}\APACrefatitle {Thinking twice about sum scores} {Thinking twice about
  sum scores}.{\BBCQ}
\newblock
\APACjournalVolNumPages{Behavior Research Methods}{52}{6}{2287--2305}.
\newblock
\begin{APACrefDOI} \doi{10.3758/s13428-020-01398-0} \end{APACrefDOI}
\PrintBackRefs{\CurrentBib}

\bibitem [\protect \citeauthoryear {%
Merkle%
, Ariyo%
, Winter%
\BCBL {}\ \BBA {} Garnier-Villarreal%
}{%
Merkle%
\ \protect \BOthers {.}}{%
{\protect \APACyear {2023}}%
}]{%
merari23}
\APACinsertmetastar {%
merari23}%
\begin{APACrefauthors}%
Merkle, E\BPBI C.%
, Ariyo, O.%
, Winter, S\BPBI D.%
\BCBL {}\ \BBA {} Garnier-Villarreal, M.%
\end{APACrefauthors}%
\unskip\
\newblock
\APACrefYearMonthDay{2023}{}{}.
\newblock
{\BBOQ}\APACrefatitle {Opaque prior distributions in {Bayesian} latent variable
  models} {Opaque prior distributions in {Bayesian} latent variable
  models}.{\BBCQ}
\newblock
\APACjournalVolNumPages{Methodology}{19}{3}{228–255}.
\newblock
\begin{APACrefDOI} \doi{10.5964/meth.11167} \end{APACrefDOI}
\PrintBackRefs{\CurrentBib}

\bibitem [\protect \citeauthoryear {%
Mislevy%
}{%
Mislevy%
}{%
{\protect \APACyear {2024}}%
}]{%
mis24}
\APACinsertmetastar {%
mis24}%
\begin{APACrefauthors}%
Mislevy, R\BPBI J.%
\end{APACrefauthors}%
\unskip\
\newblock
\APACrefYearMonthDay{2024}{}{}.
\newblock
{\BBOQ}\APACrefatitle {Are Sum Scores a Great Accomplishment of Psychometrics
  or Intuitive Test Theory?} {Are sum scores a great accomplishment of
  psychometrics or intuitive test theory?}{\BBCQ}
\newblock
\APACjournalVolNumPages{Psychometrika}{89}{4}{1170–1174}.
\newblock
\begin{APACrefDOI} \doi{10.1007/s11336-024-10003-8} \end{APACrefDOI}
\PrintBackRefs{\CurrentBib}

\bibitem [\protect \citeauthoryear {%
Muth\'{e}n%
}{%
Muth\'{e}n%
}{%
{\protect \APACyear {1984}}%
}]{%
mut84}
\APACinsertmetastar {%
mut84}%
\begin{APACrefauthors}%
Muth\'{e}n, B.%
\end{APACrefauthors}%
\unskip\
\newblock
\APACrefYearMonthDay{1984}{}{}.
\newblock
{\BBOQ}\APACrefatitle {A general structural equation model with dichotomous,
  ordered categorical, and continuous latent variable indicators} {A general
  structural equation model with dichotomous, ordered categorical, and
  continuous latent variable indicators}.{\BBCQ}
\newblock
\APACjournalVolNumPages{Psychometrika}{49}{}{115--132}.
\PrintBackRefs{\CurrentBib}

\bibitem [\protect \citeauthoryear {%
{R Core Team}%
}{%
{R Core Team}%
}{%
{\protect \APACyear {2025}}%
}]{%
r23}
\APACinsertmetastar {%
r23}%
\begin{APACrefauthors}%
{R Core Team}.%
\end{APACrefauthors}%
\unskip\
\newblock
\APACrefYearMonthDay{2025}{}{}.
\newblock
{\BBOQ}\APACrefatitle {{R: A} Language and Environment for Statistical
  Computing} {{R: A} language and environment for statistical
  computing}{\BBCQ}\ [\bibcomputersoftwaremanual].
\newblock
\APACaddressPublisher{Vienna, Austria}{}.
\newblock
\begin{APACrefURL} \url{https://www.R-project.org/} \end{APACrefURL}
\PrintBackRefs{\CurrentBib}

\bibitem [\protect \citeauthoryear {%
Reif%
\ \BBA {} Melich%
}{%
Reif%
\ \BBA {} Melich%
}{%
{\protect \APACyear {2015}}%
}]{%
reif15}
\APACinsertmetastar {%
reif15}%
\begin{APACrefauthors}%
Reif, K.%
\BCBT {}\ \BBA {} Melich, A.%
\end{APACrefauthors}%
\unskip\
\newblock
\APACrefYearMonthDay{2015}{}{}.
\newblock
\APACrefbtitle {Euro-barometer 38.1: {Consumer} Protection and Perceptions of
  Science and Technology, {November} 1992.} {Euro-barometer 38.1: {Consumer}
  protection and perceptions of science and technology, {November} 1992.}
\newblock
\APACaddressPublisher{}{Zentralarchiv f\"{u}r Empirische Sozialforschung
  [distributor], Inter-university Consortium for Political and Social Research
  [distributor]}.
\newblock
\begin{APACrefDOI} \doi{10.3886/ICPSR06045.v2} \end{APACrefDOI}
\PrintBackRefs{\CurrentBib}

\bibitem [\protect \citeauthoryear {%
Rhemtulla%
, Brosseau-Liard%
\BCBL {}\ \BBA {} Savalei%
}{%
Rhemtulla%
\ \protect \BOthers {.}}{%
{\protect \APACyear {2012}}%
}]{%
rhebro12}
\APACinsertmetastar {%
rhebro12}%
\begin{APACrefauthors}%
Rhemtulla, M.%
, Brosseau-Liard, P\BPBI E.%
\BCBL {}\ \BBA {} Savalei, V.%
\end{APACrefauthors}%
\unskip\
\newblock
\APACrefYearMonthDay{2012}{}{}.
\newblock
{\BBOQ}\APACrefatitle {When can categorical variables be treated as continuous?
  {A} comparison of robust continuous and categorical {SEM} estimation methods
  under suboptimal conditions.} {When can categorical variables be treated as
  continuous? {A} comparison of robust continuous and categorical {SEM}
  estimation methods under suboptimal conditions.}{\BBCQ}
\newblock
\APACjournalVolNumPages{Psychological Methods}{17}{3}{354–373}.
\newblock
\begin{APACrefDOI} \doi{10.1037/a0029315} \end{APACrefDOI}
\PrintBackRefs{\CurrentBib}

\bibitem [\protect \citeauthoryear {%
Rhemtulla%
\ \BBA {} Savalei%
}{%
Rhemtulla%
\ \BBA {} Savalei%
}{%
{\protect \APACyear {2025}}%
}]{%
rhesav25}
\APACinsertmetastar {%
rhesav25}%
\begin{APACrefauthors}%
Rhemtulla, M.%
\BCBT {}\ \BBA {} Savalei, V.%
\end{APACrefauthors}%
\unskip\
\newblock
\APACrefYearMonthDay{2025}{}{}.
\newblock
{\BBOQ}\APACrefatitle {Estimated Factor Scores Are Not True Factor Scores}
  {Estimated factor scores are not true factor scores}.{\BBCQ}
\newblock
\APACjournalVolNumPages{Multivariate Behavioral Research}{60}{3}{598–619}.
\newblock
\begin{APACrefDOI} \doi{10.1080/00273171.2024.2444943} \end{APACrefDOI}
\PrintBackRefs{\CurrentBib}

\bibitem [\protect \citeauthoryear {%
Rizopoulos%
}{%
Rizopoulos%
}{%
{\protect \APACyear {2006}}%
}]{%
ltm}
\APACinsertmetastar {%
ltm}%
\begin{APACrefauthors}%
Rizopoulos, D.%
\end{APACrefauthors}%
\unskip\
\newblock
\APACrefYearMonthDay{2006}{}{}.
\newblock
{\BBOQ}\APACrefatitle {{\em ltm}: An {R} package for Latent Variable Modelling
  and Item Response Theory Analyses} {{\em ltm}: An {R} package for latent
  variable modelling and item response theory analyses}.{\BBCQ}
\newblock
\APACjournalVolNumPages{Journal of Statistical Software}{17}{5}{1--25}.
\newblock
\begin{APACrefURL} \url{https://doi.org/10.18637/jss.v017.i05} \end{APACrefURL}
\PrintBackRefs{\CurrentBib}

\bibitem [\protect \citeauthoryear {%
Rosseel%
}{%
Rosseel%
}{%
{\protect \APACyear {2012}}%
}]{%
ros12}
\APACinsertmetastar {%
ros12}%
\begin{APACrefauthors}%
Rosseel, Y.%
\end{APACrefauthors}%
\unskip\
\newblock
\APACrefYearMonthDay{2012}{}{}.
\newblock
{\BBOQ}\APACrefatitle {{\em lavaan}: {An} {R} Package for Structural Equation
  Modeling} {{\em lavaan}: {An} {R} package for structural equation
  modeling}.{\BBCQ}
\newblock
\APACjournalVolNumPages{Journal of Statistical Software}{48}{2}{1--36}.
\newblock
\begin{APACrefURL} \url{https://doi.org/10.18637/jss.v048.i02} \end{APACrefURL}
\PrintBackRefs{\CurrentBib}

\bibitem [\protect \citeauthoryear {%
Rosseel%
}{%
Rosseel%
}{%
{\protect \APACyear {2015}}%
}]{%
ros15}
\APACinsertmetastar {%
ros15}%
\begin{APACrefauthors}%
Rosseel, Y.%
\end{APACrefauthors}%
\unskip\
\newblock
\APACrefYearMonthDay{2015}{}{}.
\newblock
{\BBOQ}\APACrefatitle {{\em lavaan} note: {Equality} constraints} {{\em lavaan}
  note: {Equality} constraints}.{\BBCQ}
\newblock
\APACjournalVolNumPages{Obtained from
  \url{https://lavaan.ugent.be/notes/lavaan_eq_constraints.pdf}}{}{}{}.
\PrintBackRefs{\CurrentBib}

\bibitem [\protect \citeauthoryear {%
Samejima%
}{%
Samejima%
}{%
{\protect \APACyear {1969}}%
}]{%
sam69}
\APACinsertmetastar {%
sam69}%
\begin{APACrefauthors}%
Samejima, F.%
\end{APACrefauthors}%
\unskip\
\newblock
\APACrefYearMonthDay{1969}{}{}.
\newblock
\APACrefbtitle {Estimation of latent ability using a response pattern of graded
  scores} {Estimation of latent ability using a response pattern of graded
  scores}\ \APACbVolEdTR{}{\BTR{}\ \BNUM~17}.
\newblock
\APACaddressInstitution{}{Psychometrika Monograph Supplement}.
\PrintBackRefs{\CurrentBib}

\bibitem [\protect \citeauthoryear {%
Shi%
\ \BBA {} Lee%
}{%
Shi%
\ \BBA {} Lee%
}{%
{\protect \APACyear {1998}}%
}]{%
shi98}
\APACinsertmetastar {%
shi98}%
\begin{APACrefauthors}%
Shi, J\BHBI Q.%
\BCBT {}\ \BBA {} Lee, S\BHBI Y.%
\end{APACrefauthors}%
\unskip\
\newblock
\APACrefYearMonthDay{1998}{}{}.
\newblock
{\BBOQ}\APACrefatitle {Bayesian sampling-based approach for factor analysis
  models with continuous and polytomous data} {Bayesian sampling-based approach
  for factor analysis models with continuous and polytomous data}.{\BBCQ}
\newblock
\APACjournalVolNumPages{British Journal of Mathematical and Statistical
  Psychology}{51}{2}{233--252}.
\PrintBackRefs{\CurrentBib}

\bibitem [\protect \citeauthoryear {%
Sijtsma%
, Ellis%
\BCBL {}\ \BBA {} Borsboom%
}{%
Sijtsma%
\ \protect \BOthers {.}}{%
{\protect \APACyear {2024}}%
{\protect \APACexlab {{\protect \BCnt {1}}}}}]{%
sij24}
\APACinsertmetastar {%
sij24}%
\begin{APACrefauthors}%
Sijtsma, K.%
, Ellis, J\BPBI L.%
\BCBL {}\ \BBA {} Borsboom, D.%
\end{APACrefauthors}%
\unskip\
\newblock
\APACrefYearMonthDay{2024{\protect \BCnt {1}}}{}{}.
\newblock
{\BBOQ}\APACrefatitle {Recognize the Value of the Sum Score, Psychometrics’
  Greatest Accomplishment} {Recognize the value of the sum score,
  psychometrics’ greatest accomplishment}.{\BBCQ}
\newblock
\APACjournalVolNumPages{Psychometrika}{89}{1}{84--117}.
\newblock
\begin{APACrefDOI} \doi{10.1007/s11336-024-09964-7} \end{APACrefDOI}
\PrintBackRefs{\CurrentBib}

\bibitem [\protect \citeauthoryear {%
Sijtsma%
, Ellis%
\BCBL {}\ \BBA {} Borsboom%
}{%
Sijtsma%
\ \protect \BOthers {.}}{%
{\protect \APACyear {2024}}%
{\protect \APACexlab {{\protect \BCnt {2}}}}}]{%
sijell24}
\APACinsertmetastar {%
sijell24}%
\begin{APACrefauthors}%
Sijtsma, K.%
, Ellis, J\BPBI L.%
\BCBL {}\ \BBA {} Borsboom, D.%
\end{APACrefauthors}%
\unskip\
\newblock
\APACrefYearMonthDay{2024{\protect \BCnt {2}}}{}{}.
\newblock
{\BBOQ}\APACrefatitle {Rejoinder to {McNeish and Mislevy}: {What} Does
  Psychological Measurement Require?} {Rejoinder to {McNeish and Mislevy}:
  {What} does psychological measurement require?}{\BBCQ}
\newblock
\APACjournalVolNumPages{Psychometrika}{89}{4}{1175–1185}.
\newblock
\begin{APACrefDOI} \doi{10.1007/s11336-024-10004-7} \end{APACrefDOI}
\PrintBackRefs{\CurrentBib}

\bibitem [\protect \citeauthoryear {%
Steiger%
}{%
Steiger%
}{%
{\protect \APACyear {1994}}%
}]{%
ste94}
\APACinsertmetastar {%
ste94}%
\begin{APACrefauthors}%
Steiger, J\BPBI H.%
\end{APACrefauthors}%
\unskip\
\newblock
\APACrefYearMonthDay{1994}{}{}.
\newblock
{\BBOQ}\APACrefatitle {Factor analysis in the 1980’s and the 1990’s: Some
  old debates and some new developments} {Factor analysis in the 1980’s and
  the 1990’s: Some old debates and some new developments}.{\BBCQ}
\newblock
\BIn{} I.~Borg\ \BBA {} P\BPBI P.~Mohler\ (\BEDS), \APACrefbtitle {Trends and
  Perspectives in Empirical Social Research} {Trends and perspectives in
  empirical social research}\ (\BPGS\ 201--224).
\newblock
\APACaddressPublisher{}{De Gruyter}.
\PrintBackRefs{\CurrentBib}

\bibitem [\protect \citeauthoryear {%
Steiger%
}{%
Steiger%
}{%
{\protect \APACyear {2002}}%
}]{%
ste02}
\APACinsertmetastar {%
ste02}%
\begin{APACrefauthors}%
Steiger, J\BPBI H.%
\end{APACrefauthors}%
\unskip\
\newblock
\APACrefYearMonthDay{2002}{}{}.
\newblock
{\BBOQ}\APACrefatitle {When constraints interact: {A} caution about reference
  variables, identification constraints, and scale dependencies in structural
  equation modeling.} {When constraints interact: {A} caution about reference
  variables, identification constraints, and scale dependencies in structural
  equation modeling.}{\BBCQ}
\newblock
\APACjournalVolNumPages{Psychological Methods}{7}{2}{210–227}.
\newblock
\begin{APACrefDOI} \doi{10.1037/1082-989x.7.2.210} \end{APACrefDOI}
\PrintBackRefs{\CurrentBib}

\bibitem [\protect \citeauthoryear {%
Stevens%
}{%
Stevens%
}{%
{\protect \APACyear {1946}}%
}]{%
ste46}
\APACinsertmetastar {%
ste46}%
\begin{APACrefauthors}%
Stevens, S\BPBI S.%
\end{APACrefauthors}%
\unskip\
\newblock
\APACrefYearMonthDay{1946}{}{}.
\newblock
{\BBOQ}\APACrefatitle {On the theory of scales of measurement} {On the theory
  of scales of measurement}.{\BBCQ}
\newblock
\APACjournalVolNumPages{Science}{103}{2684}{677--680}.
\PrintBackRefs{\CurrentBib}

\bibitem [\protect \citeauthoryear {%
Takane%
\ \BBA {} de Leeuw%
}{%
Takane%
\ \BBA {} de Leeuw%
}{%
{\protect \APACyear {1987}}%
}]{%
takdel87}
\APACinsertmetastar {%
takdel87}%
\begin{APACrefauthors}%
Takane, Y.%
\BCBT {}\ \BBA {} de Leeuw, J.%
\end{APACrefauthors}%
\unskip\
\newblock
\APACrefYearMonthDay{1987}{}{}.
\newblock
{\BBOQ}\APACrefatitle {On the relationship between item response theory and
  factor analysis of discretized variables} {On the relationship between item
  response theory and factor analysis of discretized variables}.{\BBCQ}
\newblock
\APACjournalVolNumPages{Psychometrika}{52}{}{393--408}.
\PrintBackRefs{\CurrentBib}

\bibitem [\protect \citeauthoryear {%
Tuerlinckx%
, Rijmen%
, Verbeke%
\BCBL {}\ \BBA {} De~Boeck%
}{%
Tuerlinckx%
\ \protect \BOthers {.}}{%
{\protect \APACyear {2006}}%
}]{%
tuerij06}
\APACinsertmetastar {%
tuerij06}%
\begin{APACrefauthors}%
Tuerlinckx, F.%
, Rijmen, F.%
, Verbeke, G.%
\BCBL {}\ \BBA {} De~Boeck, P.%
\end{APACrefauthors}%
\unskip\
\newblock
\APACrefYearMonthDay{2006}{}{}.
\newblock
{\BBOQ}\APACrefatitle {Statistical inference in generalized linear mixed
  models: {A} review} {Statistical inference in generalized linear mixed
  models: {A} review}.{\BBCQ}
\newblock
\APACjournalVolNumPages{British Journal of Mathematical and Statistical
  Psychology}{59}{2}{225--255}.
\newblock
\begin{APACrefDOI} \doi{10.1348/000711005x79857} \end{APACrefDOI}
\PrintBackRefs{\CurrentBib}

\bibitem [\protect \citeauthoryear {%
Waller%
}{%
Waller%
}{%
{\protect \APACyear {2022}}%
}]{%
wal22}
\APACinsertmetastar {%
wal22}%
\begin{APACrefauthors}%
Waller, N\BPBI G.%
\end{APACrefauthors}%
\unskip\
\newblock
\APACrefYearMonthDay{2022}{}{}.
\newblock
{\BBOQ}\APACrefatitle {Breaking Our Silence on Factor Score Indeterminacy}
  {Breaking our silence on factor score indeterminacy}.{\BBCQ}
\newblock
\APACjournalVolNumPages{Journal of Educational and Behavioral
  Statistics}{48}{2}{244–261}.
\newblock
\begin{APACrefDOI} \doi{10.3102/10769986221128810} \end{APACrefDOI}
\PrintBackRefs{\CurrentBib}

\bibitem [\protect \citeauthoryear {%
Wickham%
}{%
Wickham%
}{%
{\protect \APACyear {2016}}%
}]{%
ggplot2}
\APACinsertmetastar {%
ggplot2}%
\begin{APACrefauthors}%
Wickham, H.%
\end{APACrefauthors}%
\unskip\
\newblock
\APACrefYear{2016}.
\newblock
\APACrefbtitle {{ggplot2: Elegant} Graphics for Data Analysis} {{ggplot2:
  Elegant} graphics for data analysis}.
\newblock
\APACaddressPublisher{}{Springer-Verlag New York}.
\newblock
\begin{APACrefURL} \url{https://ggplot2.tidyverse.org} \end{APACrefURL}
\PrintBackRefs{\CurrentBib}

\bibitem [\protect \citeauthoryear {%
Winship%
\ \BBA {} Mare%
}{%
Winship%
\ \BBA {} Mare%
}{%
{\protect \APACyear {1984}}%
}]{%
winmar84}
\APACinsertmetastar {%
winmar84}%
\begin{APACrefauthors}%
Winship, C.%
\BCBT {}\ \BBA {} Mare, R\BPBI D.%
\end{APACrefauthors}%
\unskip\
\newblock
\APACrefYearMonthDay{1984}{}{}.
\newblock
{\BBOQ}\APACrefatitle {Regression models with ordinal variables} {Regression
  models with ordinal variables}.{\BBCQ}
\newblock
\APACjournalVolNumPages{American Sociological Review}{49}{}{512--525}.
\PrintBackRefs{\CurrentBib}

\bibitem [\protect \citeauthoryear {%
Wu%
\ \BBA {} Estabrook%
}{%
Wu%
\ \BBA {} Estabrook%
}{%
{\protect \APACyear {2016}}%
}]{%
wuest16}
\APACinsertmetastar {%
wuest16}%
\begin{APACrefauthors}%
Wu, H.%
\BCBT {}\ \BBA {} Estabrook, R.%
\end{APACrefauthors}%
\unskip\
\newblock
\APACrefYearMonthDay{2016}{}{}.
\newblock
{\BBOQ}\APACrefatitle {Identification of confirmatory factor analysis models of
  different levels of invariance for ordered categorical outcomes}
  {Identification of confirmatory factor analysis models of different levels of
  invariance for ordered categorical outcomes}.{\BBCQ}
\newblock
\APACjournalVolNumPages{Psychometrika}{81}{4}{1014--1045}.
\PrintBackRefs{\CurrentBib}

\end{thebibliography}

\appendix

\section{Alternative Integer Constraints}
We considered many variations of the integer constraints presented in this paper. Our proposed constraints appear to get us close to the integer-coded average across a variety of scenarios. Below, we describe some variations that we considered, along with problems that we encountered.

\paragraph{Constrain one threshold per observed variable instead of two.} We initially constrained each observed variable's middle threshold to the middle of the integer scale. For example, for an observed variable $j$ with $K_j = 4$, we fixed $\tau_{j2} = 2.5$. For an observed variable $j$ with $K_j = 5$, we constrained the second and third thresholds to be symmetric around 3, i.e., $3 - \tau_{j2} = \tau_{j3} - 3$. Then, to make up for the lack of constraints on the second threshold, we either constrained $\theta_{jj} = 1$ or $\nu_j = 0$. A problem with these constraints is that the outer thresholds $\tau_{j1}$ and $\tau_{j(K_j-1)}$ were often estimated to be outside of $(1,K_j)$, so that there was a larger mismatch between observed averages and latent variable predictions.

\paragraph{Constrain two middle thresholds instead of two outer thresholds.} We again observed that the outer thresholds would be estimated outside of $(1,K_j)$, leading again to a larger mismatch between observed averages and latent variable predictions.

\paragraph{Constrain the loadings' geometric mean to equal 1, instead of the arithmetic mean.} Under this constraint, we may experience inconsistent signs of individual loadings. For example, when we have an even number of loadings, the loadings can all be negative yet have a geometric mean of 1. This will not necessarily be problematic in situations where all loadings are expected to have the same sign, and where we use frequentist estimation. It may be more problematic for Bayesian models where some loadings' posterior distributions overlap with 0. In that case, MCMC methods may experience problems because they will sometimes encounter loadings with differing signs.

\section{Supplemental Figures and Tables}

The figures and tables below show results from additional simulation conditions. Specifically, Figure~\ref{fig:simresdef}, Table~\ref{tab:chisqdef}, and Table~\ref{tab:chisq2def} mimic the simulation study results presented in the main text of this paper but focus on conditions when default starting values are used. In addition, Table~\ref{tab:chisq_supp} and Table~\ref{tab:chisq_supp_def} present supplementary results on the proportion of replications with identical model fit across identification constraints for conditions using simple and default starting values, respectively. Similarly, Table~\ref{tab:supp2} and Table~\ref{tab:supp2_def} present supplementary results for models with three indicators on the proportion of replications that resulted in best fit for each identification constraint for conditions using simple and default starting values, respectively. Finally, Table~\ref{tab:supp3} and Table~\ref{tab:supp3_def} present supplementary results for models with six indicators on the proportion of replications that resulted in best fit for each identification constraint for conditions using simple and default starting values, respectively.

\begin{figure}
  \caption{Proportion of converged and admissible replications across simulation conditions when all indicators have a balanced, skewed, or middling response distribution, using default starting values.}
  \label{fig:simresdef}
\begin{knitrout}\footnotesize
\definecolor{shadecolor}{rgb}{0.969, 0.969, 0.969}\color{fgcolor}

{\centering \includegraphics[width=6in,height=6in]{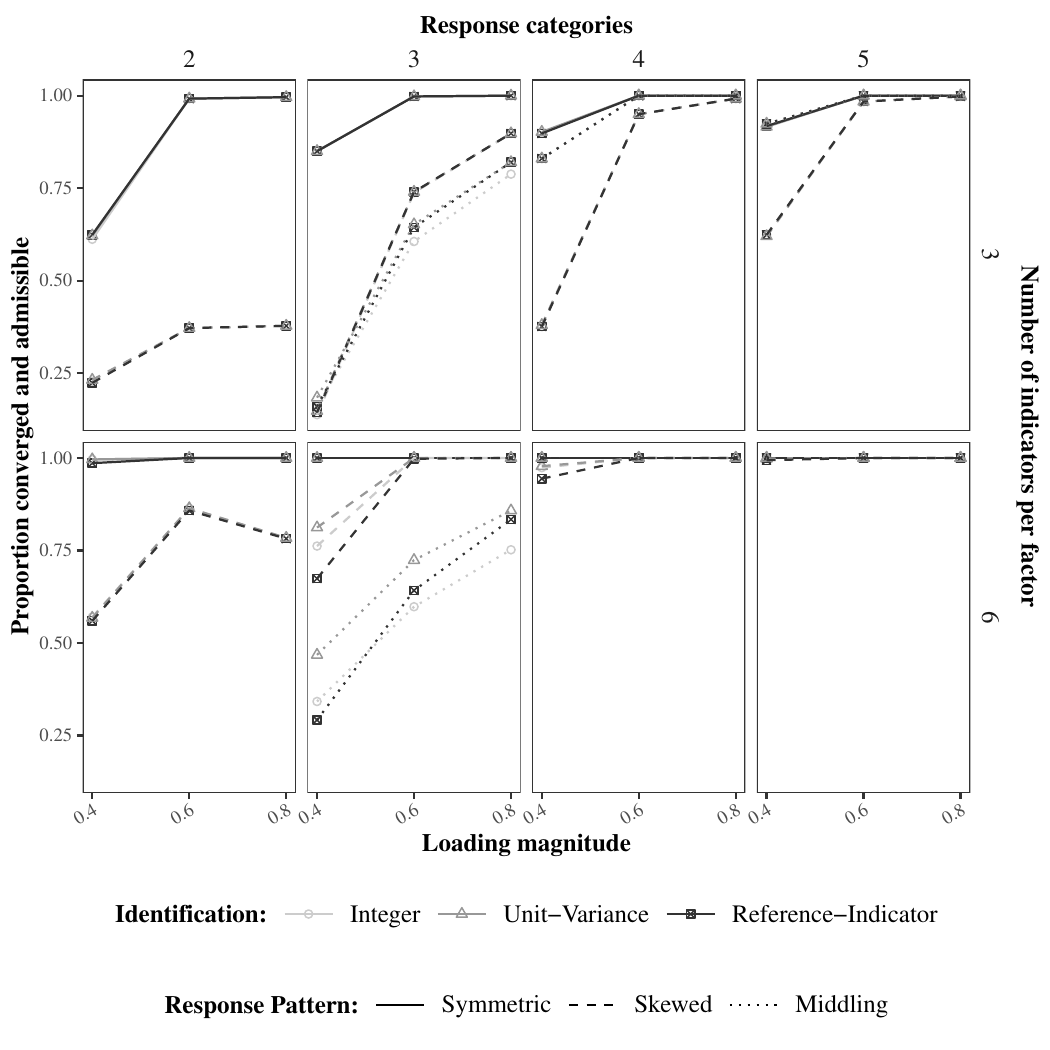} 

}

\end{knitrout}
\end{figure}

\begin{table}[ht]
\centering
\begingroup\footnotesize
\begin{tabular}{llcccccc}
  \hline
  Loading Mag. & Prop. Sparse &  \multicolumn{6}{c}{Response Options}\\ & & \multicolumn{3}{c}{3 Indicators} & \multicolumn{3}{c}{6 Indicators}\\ & & 3 & 4 & 5 & 3 & 4 & 5\\ \hline
0.4 & 0.33 & 0.99 & 1.00 & 1.00 & 0.92 & 0.99 & 0.98 \\ 
   & 0.67 & 0.98 & 1.00 & 1.00 & 0.75 & 0.98 & 0.99 \\ 
   & 1.00 & 0.92 & 1.00 & 1.00 & 0.56 & 0.99 & 0.99 \\ 
  0.6 & 0.33 & 0.99 & 1.00 & 1.00 & 0.96 & 1.00 & 1.00 \\ 
   & 0.67 & 0.98 & 1.00 & 1.00 & 0.86 & 1.00 & 1.00 \\ 
   & 1.00 & 0.96 & 1.00 & 1.00 & 0.73 & 1.00 & 1.00 \\ 
  0.8 & 0.33 & 1.00 & 1.00 & 1.00 & 0.99 & 1.00 & 1.00 \\ 
   & 0.67 & 1.00 & 1.00 & 1.00 & 0.97 & 1.00 & 1.00 \\ 
   & 1.00 & 0.99 & 1.00 & 1.00 & 0.96 & 1.00 & 1.00 \\ 
   \hline
\end{tabular}
\endgroup
\caption{Proportion replications with middling response pattern resulting in identical fit across identification constraint methods, using default starting values.} 
\label{tab:chisqdef}
\end{table}

\begin{table}[ht]
\centering
\begingroup\footnotesize
\begin{tabular}{lccccccccc}
  \hline
  Best Fit & \multicolumn{9}{c}{Proportion Sparse Indicators}\\ & \multicolumn{3}{c}{Loadings: 0.4} & \multicolumn{3}{c}{Loadings: 0.6} & \multicolumn{3}{c}{Loadings: 0.8}\\ & 0.33 & 0.67 & 1.00 & 0.33 & 0.67 & 1.00 & 0.33 & 0.67 & 1.00\\ \hline
All & 0.92 & 0.75 & 0.56 & 0.96 & 0.86 & 0.73 & 0.99 & 0.97 & 0.96 \\ 
  Reference-Indicator & 0.01 & 0.02 & 0.06 & 0.00 & 0.00 & 0.01 & 0.00 & 0.00 & 0.00 \\ 
  Unit-Variance & 0.00 & 0.03 & 0.05 & 0.00 & 0.01 & 0.00 & 0.00 & 0.00 & 0.00 \\ 
  Integer & 0.02 & 0.07 & 0.17 & 0.02 & 0.07 & 0.13 & 0.00 & 0.01 & 0.01 \\ 
  RI \& UV & 0.02 & 0.05 & 0.10 & 0.01 & 0.05 & 0.07 & 0.01 & 0.02 & 0.02 \\ 
  RI \& I & 0.00 & 0.04 & 0.04 & 0.00 & 0.01 & 0.03 & 0.00 & 0.00 & 0.00 \\ 
  UV \& I & 0.03 & 0.04 & 0.04 & 0.00 & 0.01 & 0.02 & 0.00 & 0.00 & 0.00 \\ 
   \hline
\end{tabular}
\endgroup
\caption{Proportion replications with middling response pattern, six indicators, and three response categories resulting in best fit across identification constraint methods, using default starting values.} 
\label{tab:chisq2def}
\end{table}

\begin{table}
\centering
\begingroup\footnotesize
\begin{tabular}{lllcccccc}
  \hline
  Pattern & Loading Mag. & Prop. Sparse & \multicolumn{6}{c}{Response Options}\\ & & & \multicolumn{3}{c}{3 Indicators} & \multicolumn{3}{c}{6 Indicators}\\ & & &  3 & 4 & 5 &  3 & 4 & 5\\ \hline
Symmetric & 0.4 & 0.00 & 1.00 & 1.00 & 1.00 & 0.99 & 1.00 & 0.98 \\ 
   & 0.6 & 0.00 & 1.00 & 1.00 & 1.00 & 1.00 & 1.00 & 1.00 \\ 
   & 0.8 & 0.00 & 1.00 & 1.00 & 1.00 & 1.00 & 1.00 & 1.00 \\ 
  Skewed & 0.4 & 0.33 & 0.99 & 1.00 & 0.99 & 0.96 & 0.99 & 0.99 \\ 
   &  & 0.67 & 1.00 & 1.00 & 0.99 & 0.89 & 0.99 & 0.99 \\ 
   &  & 1.00 & 0.98 & 1.00 & 1.00 & 0.92 & 0.98 & 0.99 \\ 
   & 0.6 & 0.33 & 1.00 & 1.00 & 1.00 & 0.97 & 0.99 & 1.00 \\ 
   &  & 0.67 & 1.00 & 1.00 & 1.00 & 0.96 & 0.99 & 0.99 \\ 
   &  & 1.00 & 0.99 & 1.00 & 1.00 & 0.95 & 1.00 & 1.00 \\ 
   & 0.8 & 0.33 & 1.00 & 1.00 & 1.00 & 0.96 & 0.99 & 0.99 \\ 
   &  & 0.67 & 1.00 & 1.00 & 1.00 & 0.94 & 0.99 & 0.99 \\ 
   &  & 1.00 & 1.00 & 1.00 & 1.00 & 0.97 & 0.99 & 1.00 \\ 
   \hline
\end{tabular}
\endgroup
\caption{Proportion replications with symmetric or skewed response pattern resulting in identical fit across identification constraint methods, using simple starting values.} 
\label{tab:chisq_supp}
\end{table}

\begin{table}
\centering
\begingroup\footnotesize
\begin{tabular}{lllcccccc}
  \hline
  Pattern & Loading Mag. & Prop. Sparse & \multicolumn{6}{c}{Response Options}\\ & & & \multicolumn{3}{c}{3 Indicators} & \multicolumn{3}{c}{6 Indicators}\\ & & &  3 & 4 & 5 &  3 & 4 & 5\\ \hline
Symmetric & 0.4 & 0.00 & 1.00 & 1.00 & 1.00 & 0.99 & 1.00 & 0.99 \\ 
   & 0.6 & 0.00 & 1.00 & 1.00 & 1.00 & 1.00 & 1.00 & 1.00 \\ 
   & 0.8 & 0.00 & 1.00 & 1.00 & 1.00 & 1.00 & 1.00 & 1.00 \\ 
  Skewed & 0.4 & 0.33 & 0.99 & 1.00 & 1.00 & 0.94 & 0.99 & 0.99 \\ 
   &  & 0.67 & 1.00 & 1.00 & 0.99 & 0.91 & 0.99 & 0.99 \\ 
   &  & 1.00 & 0.98 & 0.99 & 1.00 & 0.93 & 0.98 & 1.00 \\ 
   & 0.6 & 0.33 & 1.00 & 1.00 & 1.00 & 0.96 & 1.00 & 1.00 \\ 
   &  & 0.67 & 0.99 & 1.00 & 1.00 & 0.94 & 1.00 & 0.99 \\ 
   &  & 1.00 & 0.98 & 1.00 & 1.00 & 0.95 & 1.00 & 1.00 \\ 
   & 0.8 & 0.33 & 1.00 & 1.00 & 1.00 & 0.98 & 0.99 & 1.00 \\ 
   &  & 0.67 & 1.00 & 1.00 & 1.00 & 0.94 & 0.99 & 0.99 \\ 
   &  & 1.00 & 1.00 & 1.00 & 1.00 & 0.96 & 1.00 & 1.00 \\ 
   \hline
\end{tabular}
\endgroup
\caption{Proportion replications with symmetric or skewed response pattern resulting in identical fit across identification constraint methods, using default starting values.} 
\label{tab:chisq_supp_def}
\end{table}

\begin{table}
\centering
\begingroup\scriptsize
\begin{tabular}{llccccccccc}
  \hline
  Resp.Options & Best Fit & \multicolumn{9}{c}{Proportion Sparse Indicators}\\ & & \multicolumn{3}{c}{Loadings: 0.4} & \multicolumn{3}{c}{Loadings: 0.6} & \multicolumn{3}{c}{Loadings: 0.8}\\ & & 0.33 & 0.67 & 1.00 & 0.33 & 0.67 & 1.00 & 0.33 & 0.67 & 1.00\\ \hline
3 & All & 0.99 & 0.94 & 0.93 & 1.00 & 0.98 & 0.93 & 1.00 & 1.00 & 1.00 \\ 
   & Reference-Indicator & 0.00 & 0.01 & 0.00 & 0.00 & 0.00 & 0.01 & 0.00 & 0.00 & 0.00 \\ 
   & Unit-Variance & 0.00 & 0.01 & 0.02 & 0.00 & 0.00 & 0.01 & 0.00 & 0.00 & 0.00 \\ 
   & Integer & 0.00 & 0.01 & 0.00 & 0.00 & 0.01 & 0.01 & 0.00 & 0.00 & 0.00 \\ 
   & RI \& UV & 0.00 & 0.01 & 0.02 & 0.00 & 0.00 & 0.02 & 0.00 & 0.00 & 0.00 \\ 
   & RI \& I & 0.00 & 0.01 & 0.02 & 0.00 & 0.00 & 0.02 & 0.00 & 0.00 & 0.00 \\ 
   & UV \& I & 0.00 & 0.01 & 0.02 & 0.00 & 0.00 & 0.00 & 0.00 & 0.00 & 0.00 \\ 
  4 & All & 1.00 & 1.00 & 1.00 & 1.00 & 1.00 & 1.00 & 1.00 & 1.00 & 1.00 \\ 
   & Reference-Indicator & 0.00 & 0.00 & 0.00 & 0.00 & 0.00 & 0.00 & 0.00 & 0.00 & 0.00 \\ 
   & Unit-Variance & 0.00 & 0.00 & 0.00 & 0.00 & 0.00 & 0.00 & 0.00 & 0.00 & 0.00 \\ 
   & Integer & 0.00 & 0.00 & 0.00 & 0.00 & 0.00 & 0.00 & 0.00 & 0.00 & 0.00 \\ 
   & RI \& UV & 0.00 & 0.00 & 0.00 & 0.00 & 0.00 & 0.00 & 0.00 & 0.00 & 0.00 \\ 
   & RI \& I & 0.00 & 0.00 & 0.00 & 0.00 & 0.00 & 0.00 & 0.00 & 0.00 & 0.00 \\ 
   & UV \& I & 0.00 & 0.00 & 0.00 & 0.00 & 0.00 & 0.00 & 0.00 & 0.00 & 0.00 \\ 
  5 & All & 1.00 & 1.00 & 1.00 & 1.00 & 1.00 & 1.00 & 1.00 & 1.00 & 1.00 \\ 
   & Reference-Indicator & 0.00 & 0.00 & 0.00 & 0.00 & 0.00 & 0.00 & 0.00 & 0.00 & 0.00 \\ 
   & Unit-Variance & 0.00 & 0.00 & 0.00 & 0.00 & 0.00 & 0.00 & 0.00 & 0.00 & 0.00 \\ 
   & Integer & 0.00 & 0.00 & 0.00 & 0.00 & 0.00 & 0.00 & 0.00 & 0.00 & 0.00 \\ 
   & RI \& UV & 0.00 & 0.00 & 0.00 & 0.00 & 0.00 & 0.00 & 0.00 & 0.00 & 0.00 \\ 
   & RI \& I & 0.00 & 0.00 & 0.00 & 0.00 & 0.00 & 0.00 & 0.00 & 0.00 & 0.00 \\ 
   & UV \& I & 0.00 & 0.00 & 0.00 & 0.00 & 0.00 & 0.00 & 0.00 & 0.00 & 0.00 \\ 
   \hline
\end{tabular}
\endgroup
\caption{Proportion replications with middling response pattern and three indicators per factor resulting in best fit across identification constraint methods, using simple starting values.} 
\label{tab:supp2}
\end{table}

\begin{table}
\centering
\begingroup\scriptsize
\begin{tabular}{llccccccccc}
  \hline
  Resp.Options & Best Fit & \multicolumn{9}{c}{Proportion Sparse Indicators}\\ & & \multicolumn{3}{c}{Loadings: 0.4} & \multicolumn{3}{c}{Loadings: 0.6} & \multicolumn{3}{c}{Loadings: 0.8}\\ & & 0.33 & 0.67 & 1.00 & 0.33 & 0.67 & 1.00 & 0.33 & 0.67 & 1.00\\ \hline
3 & All & 0.99 & 0.98 & 0.92 & 0.99 & 0.98 & 0.96 & 1.00 & 1.00 & 0.99 \\ 
   & Reference-Indicator & 0.00 & 0.00 & 0.00 & 0.00 & 0.00 & 0.01 & 0.00 & 0.00 & 0.00 \\ 
   & Unit-Variance & 0.01 & 0.00 & 0.02 & 0.00 & 0.00 & 0.00 & 0.00 & 0.00 & 0.00 \\ 
   & Integer & 0.00 & 0.01 & 0.02 & 0.00 & 0.00 & 0.01 & 0.00 & 0.00 & 0.00 \\ 
   & RI \& UV & 0.00 & 0.01 & 0.05 & 0.01 & 0.01 & 0.02 & 0.00 & 0.00 & 0.01 \\ 
   & RI \& I & 0.00 & 0.00 & 0.00 & 0.00 & 0.00 & 0.00 & 0.00 & 0.00 & 0.00 \\ 
   & UV \& I & 0.00 & 0.01 & 0.00 & 0.00 & 0.00 & 0.00 & 0.00 & 0.00 & 0.00 \\ 
  4 & All & 1.00 & 1.00 & 1.00 & 1.00 & 1.00 & 1.00 & 1.00 & 1.00 & 1.00 \\ 
   & Reference-Indicator & 0.00 & 0.00 & 0.00 & 0.00 & 0.00 & 0.00 & 0.00 & 0.00 & 0.00 \\ 
   & Unit-Variance & 0.00 & 0.00 & 0.00 & 0.00 & 0.00 & 0.00 & 0.00 & 0.00 & 0.00 \\ 
   & Integer & 0.00 & 0.00 & 0.00 & 0.00 & 0.00 & 0.00 & 0.00 & 0.00 & 0.00 \\ 
   & RI \& UV & 0.00 & 0.00 & 0.00 & 0.00 & 0.00 & 0.00 & 0.00 & 0.00 & 0.00 \\ 
   & RI \& I & 0.00 & 0.00 & 0.00 & 0.00 & 0.00 & 0.00 & 0.00 & 0.00 & 0.00 \\ 
   & UV \& I & 0.00 & 0.00 & 0.00 & 0.00 & 0.00 & 0.00 & 0.00 & 0.00 & 0.00 \\ 
  5 & All & 1.00 & 1.00 & 1.00 & 1.00 & 1.00 & 1.00 & 1.00 & 1.00 & 1.00 \\ 
   & Reference-Indicator & 0.00 & 0.00 & 0.00 & 0.00 & 0.00 & 0.00 & 0.00 & 0.00 & 0.00 \\ 
   & Unit-Variance & 0.00 & 0.00 & 0.00 & 0.00 & 0.00 & 0.00 & 0.00 & 0.00 & 0.00 \\ 
   & Integer & 0.00 & 0.00 & 0.00 & 0.00 & 0.00 & 0.00 & 0.00 & 0.00 & 0.00 \\ 
   & RI \& UV & 0.00 & 0.00 & 0.00 & 0.00 & 0.00 & 0.00 & 0.00 & 0.00 & 0.00 \\ 
   & RI \& I & 0.00 & 0.00 & 0.00 & 0.00 & 0.00 & 0.00 & 0.00 & 0.00 & 0.00 \\ 
   & UV \& I & 0.00 & 0.00 & 0.00 & 0.00 & 0.00 & 0.00 & 0.00 & 0.00 & 0.00 \\ 
   \hline
\end{tabular}
\endgroup
\caption{Proportion replications with middling response pattern and three indicators per factor resulting in best fit across identification constraint methods, using default starting values.} 
\label{tab:supp2_def}
\end{table}

\begin{table}
\centering
\begingroup\scriptsize
\begin{tabular}{llccccccccc}
  \hline
  Resp.Options & Best Fit & \multicolumn{9}{c}{Proportion Sparse Indicators}\\ & & \multicolumn{3}{c}{Loadings: 0.4} & \multicolumn{3}{c}{Loadings: 0.6} & \multicolumn{3}{c}{Loadings: 0.8}\\ & & 0.33 & 0.67 & 1.00 & 0.33 & 0.67 & 1.00 & 0.33 & 0.67 & 1.00\\ \hline
3 & All & 0.87 & 0.65 & 0.51 & 0.92 & 0.75 & 0.56 & 0.99 & 0.97 & 0.95 \\ 
   & Reference-Indicator & 0.01 & 0.04 & 0.18 & 0.00 & 0.02 & 0.06 & 0.00 & 0.00 & 0.01 \\ 
   & Unit-Variance & 0.03 & 0.06 & 0.02 & 0.01 & 0.05 & 0.04 & 0.00 & 0.01 & 0.01 \\ 
   & Integer & 0.03 & 0.07 & 0.04 & 0.00 & 0.04 & 0.07 & 0.01 & 0.00 & 0.00 \\ 
   & RI \& UV & 0.01 & 0.05 & 0.11 & 0.02 & 0.05 & 0.09 & 0.00 & 0.02 & 0.01 \\ 
   & RI \& I & 0.01 & 0.05 & 0.07 & 0.02 & 0.04 & 0.08 & 0.00 & 0.00 & 0.02 \\ 
   & UV \& I & 0.03 & 0.09 & 0.07 & 0.02 & 0.03 & 0.10 & 0.00 & 0.00 & 0.01 \\ 
  4 & All & 0.99 & 0.98 & 0.99 & 1.00 & 0.99 & 1.00 & 1.00 & 1.00 & 1.00 \\ 
   & Reference-Indicator & 0.00 & 0.00 & 0.00 & 0.00 & 0.00 & 0.00 & 0.00 & 0.00 & 0.00 \\ 
   & Unit-Variance & 0.00 & 0.00 & 0.00 & 0.00 & 0.00 & 0.00 & 0.00 & 0.00 & 0.00 \\ 
   & Integer & 0.00 & 0.00 & 0.00 & 0.00 & 0.00 & 0.00 & 0.00 & 0.00 & 0.00 \\ 
   & RI \& UV & 0.01 & 0.02 & 0.01 & 0.00 & 0.00 & 0.00 & 0.00 & 0.00 & 0.00 \\ 
   & RI \& I & 0.00 & 0.00 & 0.00 & 0.00 & 0.00 & 0.00 & 0.00 & 0.00 & 0.00 \\ 
   & UV \& I & 0.00 & 0.00 & 0.00 & 0.00 & 0.00 & 0.00 & 0.00 & 0.00 & 0.00 \\ 
  5 & All & 0.99 & 0.99 & 0.99 & 1.00 & 1.00 & 1.00 & 1.00 & 1.00 & 1.00 \\ 
   & Reference-Indicator & 0.00 & 0.00 & 0.00 & 0.00 & 0.00 & 0.00 & 0.00 & 0.00 & 0.00 \\ 
   & Unit-Variance & 0.00 & 0.00 & 0.00 & 0.00 & 0.00 & 0.00 & 0.00 & 0.00 & 0.00 \\ 
   & Integer & 0.00 & 0.00 & 0.00 & 0.00 & 0.00 & 0.00 & 0.00 & 0.00 & 0.00 \\ 
   & RI \& UV & 0.01 & 0.01 & 0.01 & 0.00 & 0.00 & 0.00 & 0.00 & 0.00 & 0.00 \\ 
   & RI \& I & 0.00 & 0.00 & 0.00 & 0.00 & 0.00 & 0.00 & 0.00 & 0.00 & 0.00 \\ 
   & UV \& I & 0.00 & 0.00 & 0.00 & 0.00 & 0.00 & 0.00 & 0.00 & 0.00 & 0.00 \\ 
   \hline
\end{tabular}
\endgroup
\caption{Proportion replications with middling response pattern and six indicators per factor resulting in best fit across identification constraint methods, using simple starting values.} 
\label{tab:supp3}
\end{table}

\begin{table}
\centering
\begingroup\scriptsize
\begin{tabular}{llccccccccc}
  \hline
  Resp.Options & Best Fit & \multicolumn{9}{c}{Proportion Sparse Indicators}\\ & & \multicolumn{3}{c}{Loadings: 0.4} & \multicolumn{3}{c}{Loadings: 0.6} & \multicolumn{3}{c}{Loadings: 0.8}\\ & & 0.33 & 0.67 & 1.00 & 0.33 & 0.67 & 1.00 & 0.33 & 0.67 & 1.00\\ \hline
3 & All & 0.92 & 0.75 & 0.56 & 0.96 & 0.86 & 0.73 & 0.99 & 0.97 & 0.96 \\ 
   & Reference-Indicator & 0.01 & 0.02 & 0.06 & 0.00 & 0.00 & 0.01 & 0.00 & 0.00 & 0.00 \\ 
   & Unit-Variance & 0.00 & 0.03 & 0.05 & 0.00 & 0.01 & 0.00 & 0.00 & 0.00 & 0.00 \\ 
   & Integer & 0.02 & 0.07 & 0.17 & 0.02 & 0.07 & 0.13 & 0.00 & 0.01 & 0.01 \\ 
   & RI \& UV & 0.02 & 0.05 & 0.10 & 0.01 & 0.05 & 0.07 & 0.01 & 0.02 & 0.02 \\ 
   & RI \& I & 0.00 & 0.04 & 0.04 & 0.00 & 0.01 & 0.03 & 0.00 & 0.00 & 0.00 \\ 
   & UV \& I & 0.03 & 0.04 & 0.04 & 0.00 & 0.01 & 0.02 & 0.00 & 0.00 & 0.00 \\ 
  4 & All & 0.99 & 0.98 & 0.99 & 1.00 & 1.00 & 1.00 & 1.00 & 1.00 & 1.00 \\ 
   & Reference-Indicator & 0.00 & 0.00 & 0.00 & 0.00 & 0.00 & 0.00 & 0.00 & 0.00 & 0.00 \\ 
   & Unit-Variance & 0.00 & 0.00 & 0.00 & 0.00 & 0.00 & 0.00 & 0.00 & 0.00 & 0.00 \\ 
   & Integer & 0.00 & 0.00 & 0.00 & 0.00 & 0.00 & 0.00 & 0.00 & 0.00 & 0.00 \\ 
   & RI \& UV & 0.01 & 0.02 & 0.01 & 0.00 & 0.00 & 0.00 & 0.00 & 0.00 & 0.00 \\ 
   & RI \& I & 0.00 & 0.00 & 0.00 & 0.00 & 0.00 & 0.00 & 0.00 & 0.00 & 0.00 \\ 
   & UV \& I & 0.00 & 0.00 & 0.00 & 0.00 & 0.00 & 0.00 & 0.00 & 0.00 & 0.00 \\ 
  5 & All & 0.98 & 0.99 & 0.99 & 1.00 & 1.00 & 1.00 & 1.00 & 1.00 & 1.00 \\ 
   & Reference-Indicator & 0.00 & 0.00 & 0.00 & 0.00 & 0.00 & 0.00 & 0.00 & 0.00 & 0.00 \\ 
   & Unit-Variance & 0.00 & 0.00 & 0.00 & 0.00 & 0.00 & 0.00 & 0.00 & 0.00 & 0.00 \\ 
   & Integer & 0.00 & 0.00 & 0.00 & 0.00 & 0.00 & 0.00 & 0.00 & 0.00 & 0.00 \\ 
   & RI \& UV & 0.02 & 0.01 & 0.01 & 0.00 & 0.00 & 0.00 & 0.00 & 0.00 & 0.00 \\ 
   & RI \& I & 0.00 & 0.00 & 0.00 & 0.00 & 0.00 & 0.00 & 0.00 & 0.00 & 0.00 \\ 
   & UV \& I & 0.00 & 0.00 & 0.00 & 0.00 & 0.00 & 0.00 & 0.00 & 0.00 & 0.00 \\ 
   \hline
\end{tabular}
\endgroup
\caption{Proportion replications with middling response pattern and six indicators per factor resulting in best fit across identification constraint methods, using simple starting values.} 
\label{tab:supp3_def}
\end{table}

\end{document}